\documentclass[journal, onecolumn]{IEEEtran}

\usepackage{cite}
\usepackage{version}
\usepackage{fullpage,epsfig,graphicx,psfrag,color,subfig}
\usepackage{amsmath}
\usepackage{amssymb}

\usepackage{xspace}
\usepackage{bbm}
\usepackage{mathrsfs}

\newcommand{\Bc}{\mathcal{B}}

\newcommand{\Ec}{\mathcal{E}}

\newcommand{\Rc}{\mathcal{R}}

\newcommand{\Uc}{\mathcal{U}}
\newcommand{\Vc}{\mathcal{V}}

\newcommand{\Xc}{\mathcal{X}}
\newcommand{\Yc}{\mathcal{Y}}
\newcommand{\Zc}{\mathcal{Z}}

\newcommand{\aep}{{\mathcal{T}_{\epsilon}^{(n)}}}

\newcommand{\Lh}{{\hat{L}}}
\newcommand{\Mh}{{\hat{M}}}

\newcommand{\Xh}{{\hat{X}}}

\newcommand{\Zh}{{\hat{Z}}}

\newcommand{\lh}{{\hat{l}}}
\newcommand{\mh}{{\hat{m}}}

\newcommand{\xh}{{\hat{x}}}

\newcommand{\zh}{{\hat{z}}}

\newcommand{\Lt}{{\tilde{L}}}
\newcommand{\Mt}{{\tilde{M}}}

\newcommand{\Rt}{{\tilde{R}}}

\newcommand{\lt}{{\tilde{l}}}

\def\d{\delta}
\def\e{\epsilon}

\DeclareMathOperator\E{\sf E}
\let\P\relax
\DeclareMathOperator\P{\sf P}

\def\textiid{i.i.d.\@\xspace}
\newcommand\iid{\ifmmode\text{ i.i.d. } \else \textiid \fi}

\newtheorem{lemma}{Lemma}
\newtheorem{theorem}{Theorem}

\newtheorem{corollary}{Corollary}

\begin{document}
\title{Cascade, Triangular and Two Way Source Coding with degraded side information at the second user\footnote{Material presented in part at Allerton Conference on Communication, Control and Computing 2010}}
\author{Yeow-Khiang Chia\IEEEauthorrefmark{1}, Haim Permuter\IEEEauthorrefmark{2} and Tsachy Weissman\IEEEauthorrefmark{3}
\thanks{\IEEEauthorrefmark{1} Yeow-Khiang Chia is with Stanford University, USA. Email: ykchia@stanford.edu}  \thanks{\IEEEauthorrefmark{2} Haim Permuter is with Ben Gurion University, Israel. Email haimp@bgu.ac.il}
\thanks{\IEEEauthorrefmark{3}Tsachy Weissman is with Stanford University and Technion, Israel Institute of Technology. Email: tsachy@stanford.edu}%
}

\maketitle

\begin{abstract}
We consider the Cascade and Triangular rate-distortion problems where the same side information is available at the source node and User 1, and the side information available at User 2 is a degraded version of the side information at the source node and User 1. We characterize the rate-distortion region for these problems. For the Cascade setup, we showed that, at User 1, decoding and re-binning the codeword sent by the source node for User 2 is optimum. We then extend our results to the Two way Cascade and Triangular setting, where the source node is interested in lossy reconstruction of the side information at User 2 via a rate limited link from User 2 to the source node. We characterize the rate distortion regions for these settings. Complete explicit characterizations for all settings are also given in the Quadratic Gaussian case. We conclude with two further extensions: A triangular source coding problem with a helper, and an extension of our Two Way Cascade setting in the Quadratic Gaussian case. 
\end{abstract}
\section{Introduction} \label{sect:1}
The problem of lossy source coding through a cascade was first considered by Yamamoto~\cite{Yamamoto}, where a source node (Node 0) sends a message to Node 1, which then sends a message to Node 2. Since Yamamoto's work, the cascade setting has been extended in recent years through incorporating side information at either Nodes 1 or 2. In~\cite{Vasudevan}, the authors considered the Cascade problem with side information $Y$ at Node 1 and $Z$ at Node 2, with the Markov Chain $X-Z-Y$. The authors provided inner and outer bounds for this setup and showed that the bounds coincide for the Gaussian case. In~\cite{Cuff}, the authors considered the Cascade problem where the side information is known only to the intermediate node and provided inner and outer bounds for this setup.

Of most relevance to this paper is the work in~\cite{Permuter}, where the authors considered the Cascade source coding problem with side information available at both Node 0 and Node 1 and established the rate distortion region for this setup. The Cascade setting was then extended to the Triangular source coding setting where an additional rate limited link is available from the source node to Node 2. 

\begin{figure} [!h] 
\begin{center}
\psfrag{x}{$X$}
\psfrag{y}{$Y$}
\psfrag{z}{$Z$}
\psfrag{r1}{$R_1$}
\psfrag{r2}{$R_2$}
\psfrag{x1}{$\Xh_1$}
\psfrag{x2}{$\Xh_2$}
\psfrag{n0}[l]{Node 0}
\psfrag{n1}[l]{Node 1}
\psfrag{n2}[l]{Node 2}
\scalebox{0.75}{\includegraphics{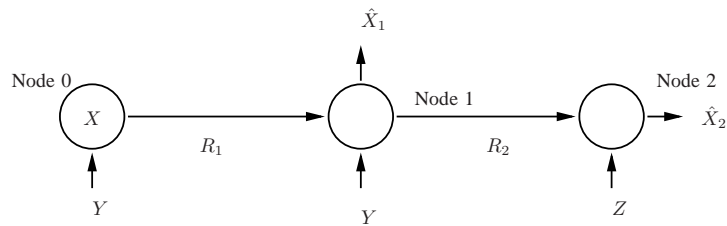}} 
\end{center}
\caption{Cascade source coding setting} \label{fig0}
\end{figure} 

\begin{figure} [!h] 
\begin{center}
\psfrag{x}{$X$}
\psfrag{y}{$Y$}
\psfrag{z}{$Z$}
\psfrag{r1}{$R_1$}
\psfrag{r2}{$R_2$}
\psfrag{r3}{$R_3$}
\psfrag{x1}{$\Xh_1$}
\psfrag{x2}{$\Xh_2$}
\psfrag{n0}[l]{Node 0}
\psfrag{n1}[l]{Node 1}
\psfrag{n2}[l]{Node 2}
\scalebox{0.75}{\includegraphics{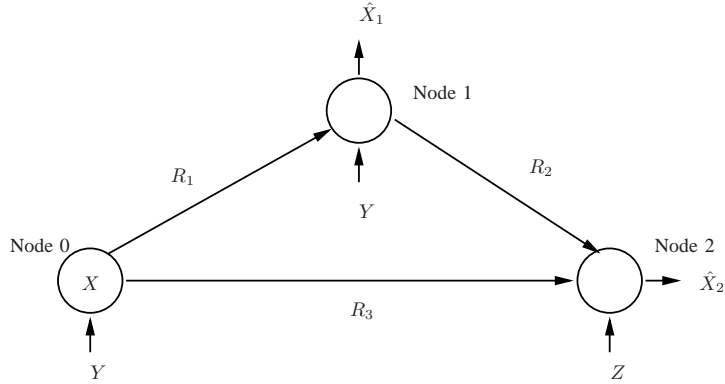}} 
\end{center}
\caption{Triangular source coding setting.} \label{fig1}
\end{figure} 

In this paper, we extend the Cascade and Triangular source coding setting in~\cite{Permuter} to include additional side information $Z$ at Node 2, with the constraint that the Markov chain $X-Y-Z$ holds. Under the Markov constraint, we establish the rate distortion regions for both the Cascade and Triangular setting. The Cascade and Triangular settings are shown in Figures~\ref{fig0} and~\ref{fig1}, respectively. In the Cascade case, we show that, at Node 1, decoding and re-binning the codeword sent by Node 0 to Node 2 is optimum. To our knowledge, this is the first setting where the decode and re-bin scheme at the Cascade is shown to be optimum. It appears to rely quite heavily on the fact that the side information at Node 2 is degraded: Since Node 1 can decode any codeword intended for Node 2, there is no need for Node 0 to send additional information for Node 1 to relay to Node 2 on the $R_1$ link. Node 0 can therefore tailor the transmission for Node 1 and rely on Node 1 to decode and minimize the rate required on the $R_2$ link. We also extend our results to two way source coding through a cascade, where Node 0 wishes to obtain a lossy version of $Z$ through a rate limited link from Node 2 to Node 0. This setup generalizes the two way source coding result found in~\cite{Kaspi}. The Two Way Cascade Source Coding and Two Way Triangular Source Coding are given in Figures~\ref{fig02} and~\ref{fig2}, respectively.

\begin{figure} [!h] 
\begin{center}
\psfrag{x}{$X$}
\psfrag{y}{$Y$}
\psfrag{z}{$Z$}
\psfrag{r1}{$R_1$}
\psfrag{r2}{$R_2$}
\psfrag{r4}{$R_3$}
\psfrag{x1}{$\Xh_1$}
\psfrag{x2}{$\Xh_2$}
\psfrag{z}{\large $Z$}
\psfrag{zh}{$\Zh$}
\psfrag{n0}[l]{Node 0}
\psfrag{n1}[l]{Node 1}
\psfrag{n2}[l]{Node 2}
\scalebox{0.75}{\includegraphics{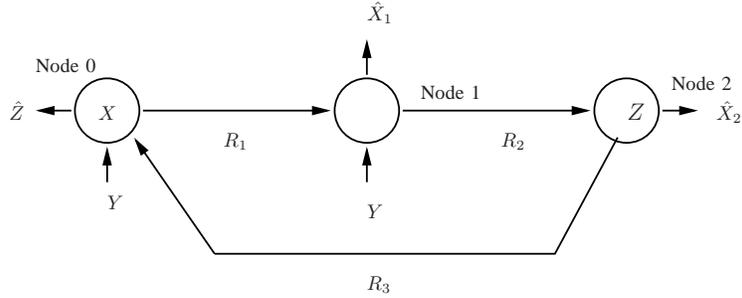}} 
\end{center}
\caption{Setup for two way cascade source coding.} \label{fig02}
\end{figure}

\begin{figure} [!h] 
\begin{center}
\psfrag{x}{$X$}
\psfrag{y}{$Y$}
\psfrag{z}{$Z$}
\psfrag{r1}{$R_1$}
\psfrag{r2}{$R_2$}
\psfrag{r3}{$R_3$}
\psfrag{r4}{$R_4$}
\psfrag{x1}{$\Xh_1$}
\psfrag{x2}{$\Xh_2$}
\psfrag{z}{\large $Z$}
\psfrag{zh}{$\Zh$}
\psfrag{n0}[l]{Node 0}
\psfrag{n1}[l]{Node 1}
\psfrag{n2}[l]{Node 2}
\scalebox{0.75}{\includegraphics{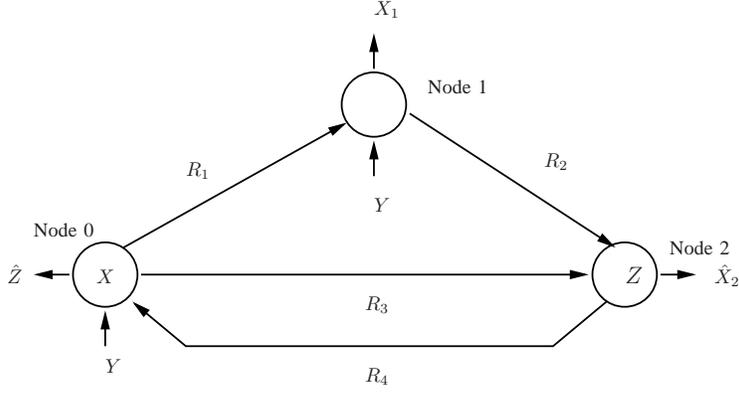}} 
\end{center}
\caption{Setup for two way triangular source coding.} \label{fig2}
\end{figure}
The rest of the paper is as follows. In section~\ref{sect:2}, we provide the formal definitions and problem setup. In section~\ref{sect:3}, we present and prove our results for the aforementioned settings. In section~\ref{sect:4}, we consider the Quadratic Gaussian case. We show that Gaussian auxiliary random variables suffice to exhaust the rate distortion regions and their parameters may be found through solving a tractable low dimensional optimization problem. We also showed that our Quadratic Gaussian settings may be transformed into equivalent settings in~\cite{Permuter} where explicit characterizations were given. In the Quadratic Gaussian case, we also extended our settings to solve a more general case of Two Way Cascade source coding. In section~\ref{sect:5}, we extend our triangular source coding setup to include a helper, which observes the side information $Y$, and has a rate limited link to Node 2. Our Two Way Cascade Quadratic Gaussian Extension is shown in Figure~\ref{fig4} (in section~\ref{sect:4}), while our helper extension is shown in Figure~\ref{fig3} (in section~\ref{sect:5}). We conclude the paper in section~\ref{sect:6}. 
\section{Problem Definition} \label{sect:2}
In this section, we give formal definitions for the setups under consideration. We will follow the notation of~\cite[Lecture 1]{El-Gamal--Kim2010}. Unless otherwise stated, all logarithms in this paper are taken to base 2. The source sequences under consideration, $\{X_i \in \Xc, i = 1, 2, \ldots\}$, $\{Y_i \in \Yc, i = 1, 2, \ldots\}$ and $\{Z_i \in \Zc, i = 1, 2, \ldots\}$, are drawn from finite alphabets $\Xc$, $\Yc$ and $\Zc$ respectively. For any $i \ge 1$, the random variables $(X_i, Y_i, Z_i)$ are independent and identically distributed according to $p(x,y,z) = p(x)p(y|x)p(z|y)$; i.e. $X- Y- Z$. The distortion measure between sequences is defined in the usual way. Let $d: \Xc \times \mathcal{\Xh} \rightarrow [0, \infty)$. Then, 
\begin{align*}
d(x^n, \xh^n) := \frac{1}{n} \sum_{i=1}^n d(x_i, \xh_i). 
\end{align*}

\subsection{Cascade and Triangular Source coding}
We give formal definition for the Triangular source coding setting (Figure~\ref{fig1}). The Cascade setting follows from specializing the definitions for the Triangular setting by setting $R_3 = 0$. A $(n, 2^{nR_1}, 2^{nR_{2}}, 2^{nR_3}, D_1, D_2)$ code for the Triangular setting consists of 3 encoders
\begin{align*}
f_1 \mbox{ (at Node 0) } &: \Xc^n \times \Yc^n \rightarrow M_1 \in [1:2^{nR_1}], \\
f_2 \mbox{ (at Node 1) } &: \Yc^n \times [1:2^{nR_1}] \rightarrow M_2 \in [1:2^{nR_2}], \\
f_3 \mbox{ (at Node 0) } &: \Xc^n \times \Yc^n \rightarrow M_3 \in [1:2^{nR_3}],
\end{align*}
and 2 decoders
\begin{align*}
g_1 \mbox{ (at Node 1) } &: \Yc^n \times [1:2^{nR_1}] \rightarrow \hat{\mathcal{X}}_1^n, \\
g_2 \mbox{ (at Node 2) } &: \Zc^n \times [1:2^{nR_2}]\times [1:2^{nR_3}] \rightarrow \hat{\mathcal{X}}_2^n,
\end{align*}
such that
\begin{align*}
\E\left[ \frac{1}{n}\sum_{i=1}^n d_j(X_i, \Xh_{j,i})\right] \le D_j, \quad \mbox{j=1,2},
\end{align*}
where $\hat{X}_1^n = g_1 (Y^n, f_1(X^n, Y^n))$ and $\hat{X}_2^n = g_2(Z^n, f_2 (Y^n, f_1(X^n, Y^n)), f_3(X^n, Y^n))$. 

Given $(D_1, D_2)$, a $(R_1, R_2, R_3)$ rate tuple for the triangular source coding setting is said to be \textit{achievable} if, for any $\e >0$ and $n$ sufficiently large, there exists a $(n, 2^{n(R_1+ \e)}, 2^{n(R_{2}+\e)}, 2^{n(R_3+ \e)}, D_1 + \e, D_2 + \e)$ code for the Triangular source coding setting.

The \textit{rate-distortion region}, $\Rc(D_1, D_2)$, is defined as the closure of the set of all achievable rate tuples.

\subsubsection*{Cascade Source coding} The Cascade source coding setting corresponds to the case where $R_3 = 0$.

\subsection{Two way Cascade and Triangular Source Coding} We give formal definitions for the more general Two way Triangular source coding setting shown in Figure~\ref{fig2}.
A \\$(n, 2^{nR_1}, 2^{nR_{2}}, 2^{nR_3}, 2^{nR_4}, D_1, D_2, D_3)$ code for the Triangular setting consists of 4 encoders
\begin{align*}
f_1 \mbox{ (at Node 0) } &: \Xc^n \times \Yc^n \rightarrow M_1 \in [1:2^{nR_1}], \\
f_2 \mbox{ (at Node 1) } &: \Yc^n \times [1:2^{nR_1}] \rightarrow M_2 \in [1:2^{nR_2}], \\
f_3 \mbox{ (at Node 0) } &: \Xc^n \times \Yc^n \rightarrow M_3 \in [1:2^{nR_3}], \\
f_4 \mbox{ (at Node 2) } &: \Zc^n \times [1:2^{nR_2}] \times [1:2^{nR_3}] \rightarrow M_4 \in [1:2^{nR_4}], \\
\end{align*}
and 3 decoders
\begin{align*}
g_1 \mbox{ (at Node 1) } &: \Yc^n \times [1:2^{nR_1}] \rightarrow \mathcal{\Xh}_1^n, \\
g_2 \mbox{ (at Node 2) } &: \Zc^n \times [1:2^{nR_2}]\times [1:2^{nR_3}] \rightarrow \mathcal{\Xh}_2^n, \\
g_3 \mbox{ (at Node 0) } &: \Xc^n \times \Yc^n \times [1:2^{nR_4}] \rightarrow \mathcal{\Zh}^n, \\
\end{align*}
such that
\begin{align*}
\E\left[ \frac{1}{n}\sum_{i=1}^n d_j(X_i, \Xh_{j,i})\right] &\le D_j, \quad \mbox{j=1,2 and}, \\
\E\left[ \frac{1}{n}\sum_{i=1}^n d_3(Z_i, \Zh_{i})\right] &\le D_j,
\end{align*}
where $\hat{X}_1^n = g_1 (Y^n, f_1(X^n, Y^n))$, $\hat{X}_2^n = g_2(Z^n, f_2 (Y^n, f_1(X^n, Y^n)), f_3(X^n, Y^n))$ and \\$\Zh^n = g_3(X^n, Y^n, f_4(Z^n, f_2 (Y^n, f_1(X^n, Y^n)), f_3(X^n, Y^n)))$.

Given $(D_1, D_2, D_3)$, a $(R_1, R_2, R_3, R_4)$ rate tuple for the two way triangular source coding setting is said to be \textit{achievable} if, for any $\e >0$ and $n$ sufficiently large, there exists a $(n, 2^{n(R_1+ \e)}, 2^{n(R_{2}+\e)}, 2^{n(R_3+ \e)},2^{n(R_4+ \e)}, D_1 + \e, D_2 + \e, D_3+ \e)$ code for the two way triangular source coding setting.

The \textit{rate-distortion region}, $\Rc(D_1, D_2, D_3)$, is defined as the closure of the set of all achievable rate tuples.

\subsubsection*{Two way Cascade Source coding} The Two way Cascade source coding setting corresponds to the case where $R_3 = 0$. In the special case of Two way Cascade setting, we will use $R_3$, rather than $R_4$, to denote the rate from Node 2 to Node 0.
\section{Main results} \label{sect:3}
In this section, we present our main results, which are single letter characterizations of the rate-distortion regions for the four settings introduced in section~\ref{sect:2}. The single letter characterizations for the Cascade source coding setting, Triangular source coding setting, Two way Cascade source coding setting and Two way Triangular source coding setting are given in Theorems~\ref{thm1},~\ref{thm2},~\ref{thm3} and~\ref{thm4}, respectively. While Theorems~\ref{thm1} to~\ref{thm3} can be derived as special cases of Theorem~\ref{thm4}, for clarity and to illustrate the development of the main ideas, we will present Theorems~\ref{thm1} to~\ref{thm4} separately. In each of the Theorems, we will present a sketch of the achievability proof and proof of the converse. Details of the achievability proofs for Theorems~\ref{thm1}-\ref{thm4} are given in Appendix~\ref{appen:1}. Proofs of the cardinality bounds for the auxiliary random variables appearing in the Theorems are given in Appendix~\ref{appen:2}.

\subsection{Cascade Source Coding}
\begin{theorem}[Rate Distortion region for Cascade source coding] \label{thm1}
$\Rc(D_1,D_2)$ for the Cascade source coding setting defined in section~\ref{sect:2} is given by the set of all rate tuples $(R_1, R_2)$ satisfying
\begin{align*}
R_2 &\ge I(U;X,Y|Z), \\
R_1 &\ge I(X;\Xh_1, U|Y)
\end{align*}
for some $p(x,y,z,u,\xh_1) = p(x)p(y|x)p(z|y)p(u|x,y)p(\xh_1|x,y,u)$ and function $g_2 : \Uc \times \Zc \rightarrow \Xh_2$ such that
\begin{align*}
\E d_j(X, \Xh_j) \le D_j, \quad \mbox{j=1,2}.
\end{align*}
\end{theorem}
The cardinality of $\mathcal{U}$ is upper bounded by $|\Uc| \le |\Xc||\Yc| + 3$.

If $Z=\emptyset$, this region reduces to the Cascade source coding region given in~\cite{Permuter}. If $Y=X$, this setup reduces to the well-known Wyner-Ziv setup~\cite{Wyner}. 

The coding scheme follows from a combination of techniques used in~\cite{Permuter} and a new idea of decoding and re-binning at the Cascade node (Node 1). Node 0 generates a description $U^n$ intended for Nodes 1 and 2. Node 1 decodes $U^n$ and then re-bins it to reduce the rate of communicating $U^n$ to Node 2 based on its side information. In addition, Node 0 generates $\Xh_1^n$ to satisfy the distortion requirement at Node 1. We now give a sketch of achievability and a proof of the converse. 

\noindent \textit{Sketch of Achievability}

We first generate $2^{n(I(X,Y;U)+\e)}$ $U^n$ sequences according to $\prod_{i=1}^n p(u_i)$. For each $u^n$ and $y^n$ sequences, we generate $2^{n(I(\Xh^n_1; X|U,Y)+\e)}$ $\Xh_1^n$ sequences according to $\prod_{i=1}^n p(\xh_i|u_i,y_i)$. Partition the set of $U^n$ sequences into $2^{n(I(U;X|Y)+2\e)}$ bins, $\Bc_1(m_{10})$. Separately and independently, partition the set of $U^n$ sequences into $2^{n(I(U;X,Y|Z)+2\e)}$ bins, $\Bc_2(m_2)$, $m_2 \in [1:2^{n(I(U;X,Y|Z)+2\e)}]$. 

Given $x^n,y^n$, Node 0 looks for a jointly typical codeword $u^n$; that is, $(u^n, x^n, y^n) \in \aep$. If there are more than one, it selects a codeword uniformly at random from the set of jointly typical codewords. This operation succeeds with high probability since there are $2^{n(I(X,Y;U)+\e)}$ $U^n$ sequences. Node 0 then looks for a $\xh_1^n$ that is jointly typical with $u^n, x^n, y^n$. This operation succeeds with high probability since there are $2^{n(I(\Xh_1; X|U,Y)+\e)}$ $\xh_1^n$ sequences. Node 0 then sends out the bin index $m_{10}$ such that $u^n \in \Bc_1(m_{10})$ and the index corresponding to $\xh_1^n$. This requires a total rate of $R_1 = I(U;X|Y)+ I(\Xh^n_1; X|U,Y)+3\e$.

At Node 1, it recovers $u^n$ by looking for the unique $u^n$ sequence in $\Bc_1(m_{10})$ such that $(u^n, y^n) \in \aep$. Since there are only $2^{n(I(X,Y;U) - I(U;X|Y)-\e)} = 2^{n(I(U;Y)-\e)}$ sequences in the bin, this operation succeeds with high probability. Node 1 reconstructs $x^n$ as $\xh_1^n$. Node 1 then sends out $m_2$ such that $u^n \in \Bc_2(m_2)$. This requires a rate of $R_2 = I(U;X,Y|Z) + 2\e$. 

At Node 2, note that since $U-(X,Y)-Z$, the sequences $(U^n, X^n, Y^n, Z^n)$ are jointly typical with high probability. Node 2 looks for the unique $u^n$ in $\Bc_2(m_2)$ such that $(u^n,z^n) \in \aep$. From the Markov Chain $U-(X,Y) - Z$, $I(U;X,Y) - I(U;X,Y|Z) = I(U;Z)$. Hence, this operation succeeds with high probability since there are only $2^{n(I(U;Z)-\e)}$ $u^n$ sequences in the bin. It then reconstructs using $\xh_i = g_2(u_i, z_i)$ for $i\in [1:n]$.

\begin{proof}[Proof of Converse]
Given a $(n, 2^{nR_1}, 2^{nR_2}, D_1, D_2)$ code, define $U_i = (X^{i-1}, Y^{i-1}, Z^{i-1}, Z_{i+1}^n, M_2)$. We have the following.{\allowdisplaybreaks
\begin{align*}
nR_2 &\ge H(M_2) \\
&\ge H(M_2|Z^n) \\
& = I(X^n, Y^n; M_2|Z^n) \\
& = \sum_{i=1}^n I(X_i, Y_i; M_2|Z^n, X^{i-1}, Y^{i-1}) \\
& = \sum_{i=1}^n H(X_i, Y_i|Z^n, X^{i-1}, Y^{i-1}) - H(X_i, Y_i|Z^n, X^{i-1}, Y^{i-1}, M_2)\\
& = \sum_{i=1}^n H(X_i, Y_i|Z_i) - H(X_i, Y_i|Z_i, U_i)\\
& = \sum_{i=1}^n I(X_i, Y_i; U_i|Z_i).
\end{align*}}
Next,{\allowdisplaybreaks
\begin{align*}
nR_1 &\ge H(M_1) \\
&\ge H(M_1|Y^n, Z^n) \\
& = H(M_1,M_2|Y^n, Z^n) = I(X^n; M_1,M_2|Y^n, Z^n) \\
& =\sum_{i=1}^n I(X_i; M_1, M_2|X^{i-1}, Y^n, Z^n) \\
& =\sum_{i=1}^n H(X_i|X^{i-1}, Y^n, Z^n) - H(X_i|X^{i-1}, Y^n, Z^n, M_1, M_2) \\
& =\sum_{i=1}^n H(X_i|Y_i, Z_i) - H(X_i|X^{i-1}, Y^n, Z^n, M_1, M_2) \\
& \stackrel{(a)}{=}\sum_{i=1}^n H(X_i|Y_i) - H(X_i|X^{i-1}, Y^n, \Xh_{1i}, Z^n, M_1, M_2) \\
& \ge \sum_{i=1}^n H(X_i|Y_i) - H(X_i|\Xh_{1i}, Y_i, U_i) \\
& = \sum_{i=1}^n I(X_i; \Xh_{1i}, U_i|Y_i).
\end{align*}}
Step (a) follows from the Markov assumption $X-Y-Z$ and the fact that $\Xh_{1i}$ is a function of $(Y^n,M_2)$. Next, let $Q$ be a random variable uniformly distributed over $[1:n]$ and independent of $(X^n, Y^n, Z^n)$. We note that $X_{Q} = X$, $Y_{Q} = Y$, $Z_Q = Z$ and{\allowdisplaybreaks
\begin{align*}
R_2 &\ge I(X_Q, Y_Q; U_Q|Q, Z_Q) \\
&= I(X_Q,Y_Q; U_Q, Q|Z_Q) \\
& = I(X,Y; U_Q,Q|Z), \\
R_1 &= I(X_Q; \Xh_{1Q}, U_Q|Y_Q, Q) \\
& = I(X; \Xh_{1Q}, U_Q, Q|Y).
\end{align*}}
Defining $U = (U_Q, Q)$ and $\Xh_{1Q} = \Xh_{1}$ then completes the proof. The existence of the reconstruction function $g_2$ follows from the definition of $U$. The Markov Chains $U - (X,Y) - Z$ and $Z - (U,X,Y) - \Xh_1$ required to factor the probability distribution stated in the Theorem also follow from definitions of $U$ and $\Xh_1$. 
\end{proof}
We now extend Theorem~\ref{thm1} to the Triangular Source coding setting.
\subsection{Triangular Source Coding}
\begin{theorem}[Rate Distortion Region for Triangular Source Coding] \label{thm2}
$\Rc(D_1,D_2)$ for the Triangular source coding setting defined in section~\ref{sect:2} is given by the set of all rate tuples $(R_1, R_2, R_3)$ satisfying
\begin{align*}
R_1 &\ge I(X;\Xh_1, U|Y), \\
R_2 &\ge I(X,Y;U|Z), \\
R_3 &\ge I(X,Y;V|U,Z)
\end{align*}
for some $p(x,y,z,u,v, \xh_1) = p(x)p(y|x)p(z|y)p(u|x,y)p(\xh_1|x,y,u)p(v|x,y,u)$ and function $g_2 : \Uc \times \Vc \times \Zc \rightarrow \Xh_2$ such that
\begin{align*}
\E d_j(X, \Xh_j) \le D_j, \quad \mbox{j=1,2}.
\end{align*}
The cardinalities for the auxiliary random variables can be upper bounded by $|\Uc|\le |\Xc||\Yc| + 4$ and $|\Vc|\le (|\Xc||\Yc| + 4)(|\Xc||\Yc| +1)$.
\end{theorem}
If $Z=\emptyset$, this region reduces to the Triangular source coding region given in~\cite{Permuter}. 

The proof of the Triangular case follows that of the Cascade case, with the additional step of Node 0 generating an additional description $V^n$ that is intended for Node 2. This description is then binned to reduce the rate, with the side information at Node 2 being $U^n$ and $Z^n$. Node 2 first decodes $U^n$ and then $V^n$.

\noindent \textit{Sketch of Achievability}

The Achievability proof is an extension of that in Theorem~\ref{thm1}. The additional step we have here is that we generate $2^{n(I(V;X,Y|U)+ \e)}$ $V^n$ sequences according to $\prod_{i=1}^n p(v_i|u_i)$ for each $u^n$ sequence, and bin these sequences to $2^{n(I(V;X,Y|U, Z)+ 2\e)}$ bins, $\Bc_3(m_3),$ $m_3 \in [1:2^{nR_3}]$. To send from Node 0 to Node 2, Node 0 first finds a $v^n$ sequence that is jointly typical with $(u^n, x^n, y^n)$. This operation succeeds with high probability since we have $2^{n(I(V;X,Y|U)+ \e)}$ $v^n$ sequences. We then send out $m_3$, the bin number for $v^n$. At Node 2, from the probability distribution, we have the Markov Chain $(V,U)- (X,Y) - Z$. Hence, the sequences are jointly typical with high probability. Node 2 reconstructs by looking for unique $v^n \in \Bc_3(m_3)$ such that $(u^n, v^n, z^n)$ are jointly typical. This operation succeeds with high probability since the number of sequences in $\Bc_3(m_3)$ is $2^{n(I(V;Z|U)-\e)}$. Node 2 then reconstructs using the function $g_2$. 

\begin{proof}[Proof of Converse]
The converse is proved in two parts. In the first part, we derive the required inequalities and in the second part, we show that the joint probability distribution can be restricted to the form stated in the Theorem.

Given a $(n, 2^{nR_1}, 2^{nR_2}, 2^{nR_3}, D_1, D_2)$ code, define $U_i = (X^{i-1}, Y^{i-1},Z^{i-1}, Z_{i+1}^n, M_2)$ and $V_i = (U_i, M_3)$. We omit proof of the $R_1$ and $R_2$ inequalities since it follows the same steps as in Theorem~\ref{thm1}. We have
\begin{align*}
nR_1 & \ge \sum_{i=1}^n I(X_i; \Xh_{1i}, U_i|Y_i), \\
nR_2 & \ge  \sum_{i=1}^n I(X_i, Y_i; U_i|Z_i).
\end{align*}
For $R_3$, we have{\allowdisplaybreaks
\begin{align*}
nR_3 &\ge H(M_3) \\
&\ge H(M_3|M_2, Z^n)\\
& = I(X^n, Y^n; M_3|M_2, Z^n) \\
& = \sum_{i=1}^n H(X_i, Y_i|M_2, Z^n, X^{i-1}, Y^{i-1}) - H(X_i, Y_i|M_2, M_3, Z^n, X^{i-1}, Y^{i-1}) \\
& = \sum_{i=1}^n H(X_i, Y_i|U_i, Z_i) - H(X_i, Y_i|U_i, V_i,Z_i) \\
& = \sum_{i=1}^n I(X_i, Y_i; V_i|U_i, Z_i).
\end{align*}}
Next, let $Q$ be a random variable uniformly distributed over $[1:n]$ and independent of $(X^n, Y^n, Z^n)$. Defining $U = (U_Q, Q)$, $V = (V_Q,Q)$ and $\Xh_{1Q} = \Xh_{1}$ then gives us the bounds stated in Theorem~\ref{thm2}. The existence of the reconstruction function $g_2$ follows from the definition of $U$ and $V$. Next, from the definitions of $U$, $V$ and $\Xh_1$, we note the following Markov relation: $(U,V,\Xh_1) - (X,Y) - Z$. The joint probability distribution can then be factored as $p(x,y,z,u,v, \xh_1) = p(x,y,z)p(u|x,y)p(\xh_1,v|x,y,u)$. 

We now show that it suffices to restrict the joint probability distributions to the form \\$p(x,y,z)p(u|x,y)p(\xh_1|x,y,u)p(v|x,y,u)$ using a method in~\cite[Lemma 5]{Permuter}. The basic idea is that since the inequalities derived rely on $p(\xh_1,v|x,y,u)$ only through the marginals $p(\xh_1|x,y,u)$ and $p(v|x,y,u)$, we can obtain the same bounds even when the probability distribution is restricted to the form $p(x,y,z)p(u|x,y)p(\xh_1|x,y,u)p(v|x,y,u)$.

Fix a joint distribution $p(x,y,z)p(u|x,y)p(\xh_1,v|x,y,u)$ and let $\hat{p}(v|x,y,u)$ and $\hat{p}(\xh_1|x,y,u)$ be the induced conditional distributions. Note that $p(x,y,z)p(u|x,y)p(\xh_1,v|x,y,u)$ and $p(x,y,z)p(u|x,y)\hat{p}(\xh_1|x,y,u)\hat{p}(v|x,y,u)$ have the same marginals $p(x,y,z,u,v)$ and $p(x,y,z,u,\xh_1)$, and the Markov condition $(U,V, \Xh_1) - (X,Y) - Z$ continues to hold under $p(x,y,z)p(u|x,y)\hat{p}(\xh_1|x,y,u)\hat{p}(v|x,y,u)$. 

Finally, note that the rate and distortion constraints given in Theorem~\ref{thm2} depends on the joint distribution only through the marginals $p(x,y,z,u,v)$ and $p(x,y,z,u,\xh_1)$. It therefore suffices to restrict the probability distributions to the form $p(x,y,z)p(u|x,y)\hat{p}(\xh_1|x,y,u)\hat{p}(v|x,y,u)$.
\end{proof}
\subsection{Two Way Cascade Source Coding}
We now extend the source coding settings to include the case where Node 0 requires a lossy version of $Z$. We first consider the Two Way Cascade Source coding setting defined in section~\ref{sect:2} (we will use $R_3$ to denote the rate on the link from Node 2 to Node 0). In the forward part, the achievable scheme consists of using the achievable scheme for the Cascade source coding case. Node 2 then sends back a description of $Z^n$ to Node 0, with $X^n, Y^n, U_1^n$ as side information at Node 0. For the converse, we rely on the techniques introduced and also on a technique for establishing Markovity of random variables found in~\cite{Kaspi}.

\begin{theorem}[Rate Distortion Region for Two Way Cascade Source Coding] \label{thm3}

$\Rc(D_1,D_2, D_3)$ for Two Way Cascade Source Coding is given by the set of all rate tuples $(R_1, R_2, R_3)$ satisfying
\begin{align*}
R_1 &\ge I(X;\Xh_1, U_{1}|Y), \\
R_2 &\ge I(U_1;X,Y|Z), \\
R_3 &\ge I(U_2;Z|U_1, X,Y),
\end{align*}
for some $p(x,y,z, u_1, u_2, \xh_1)= p(x)p(y|x)p(z|y)p(u_1|x,y)p(\xh_1|u_1,x,y)p(u_2|z,u_1)$ and functions $g_2 : \Uc_1 \times \Zc \rightarrow \Xh_2$ and $g_3 : \Uc_1 \times \Uc_2 \times \Xc \times \Yc \rightarrow \Zh$ such that
\begin{align*}
\E(d_j(X, \Xh_j)) &\le D_j, \; j = 1,2 \\
\E(d_3(Z, \Zh)) &\le D_3. 
\end{align*}
The cardinalities for the auxiliary random variables can be upper bounded by $|\Uc_1| \le |\Xc||\Yc|+5$ and $|\Uc_2| \le |\Uc_1|(|\Zc|+1)$.
\end{theorem}
If $Y=X$, this region reduces to the result for two way source coding found in~\cite{Kaspi}.

\noindent \textit{Sketch of Achievability}

The forward path ($R_1$ and $R_2$) follows from the Cascade source coding case in Theorem~\ref{thm1}. The reverse direction follows by the following. For each $u_1^n$, we generate $2^{n(I(U_2;Z|U_1)+ \e)}$ $u^n_2$ sequences according to $\prod_{i=1}^n p(u_{2i}|u_{1i})$ and bin them to $2^{n(I(U_2;Z|U_1, X, Y)+2\e)}$ bins, $\Bc_3(m_3)$, $m_3 \in [1:2^{nR_3}]$. Node 2 finds a $u_2^n$ sequence that is jointly typical with $(u_1^n, z^n)$. Since there are $2^{n(I(U_2;Z|U_1)+ \e)}$ sequences, this operation succeeds with high probability. It then sends out the bin index $m_3$, which the jointly typical $v^n$ sequence is in. At Node 0, it recovers $u_2^n$ by looking for the unique sequence in $\Bc_3(m_3)$ such that $(u_1^n, u_2^n, x^n, y^n)$ are jointly typical. From the Markov condition $U_2- (U_1, Z) - (X,Y)$ and the Markov Lemma~\cite{Tung}, the sequences are jointly typical with high probability. Next, since there are only $2^{n(I(U_2; X,Y|U_1)-\e)}$ sequences in the bin, the probability that we do not find the unique (correct) sequence goes to zero with $n$. Finally, Node 0 reconstructs using the function $g_3$.

\begin{proof}[Proof of Converse]
Given a $(n, 2^{nR_1}, 2^{nR_2}, 2^{nR_3}, D_1, D_2, D_3)$ code, define $U_{1i} = (M_2, X^{i-1}, Y^{i-1}, Z_{i+1}^n)$ and $U_{2i} =M_3$. We have{\allowdisplaybreaks
\begin{align*}
nR_1 &\ge H(M_1) \\
&\ge H(M_1 |Y^n,Z^n)\\
& =  H(M_1, M_2 |Y^n,Z^n) \\
&= I(X^n; M_1, M_2|Y^n,Z^n)\\
& = \sum_{i=1}^n I(X_i; M_1, M_2|X^{i-1}, Y^n, Z^n) \\
& = \sum_{i=1}^n H(X_i|X^{i-1}, Y^n, Z^n) - H(X_i|X^{i-1}, Y^n, Z^n, M_1, M_2) \\
& = \sum_{i=1}^n H(X_i|Y_i, Z_i) - H(X_i|X^{i-1}, Y^n, Z^n, M_1, M_2) \\
& \stackrel{(a)}{=} \sum_{i=1}^n H(X_i|Y_i) - H(X_i|X^{i-1}, Y^n, Z^n, M_1, M_2) \\
& \stackrel{(b)}{=} \sum_{i=1}^n H(X_i|Y_i) - H(X_i|X^{i-1}, \Xh_{1i}, Y^n, Z^n, M_1, M_2) \\
& {\ge} \sum_{i=1}^n H(X_i|Y_i) - H(X_i|\Xh_{1i}, Y_i, U_{1i}) \\
& = \sum_{i=1}^n I(X_i;\Xh_{1i}, U_{1i}|Y_i),
\end{align*}}
where step (a) follows from the Markov assumption $X_i - Y_i - Z_i$ and step (b) follows from $\Xh_{1i}$ being a function of $(Y^n, M_1)$.

Consider now $R_2${\allowdisplaybreaks 
\begin{align*}
nR_2 &= H(M_2) \\
& \ge H(M_2|Z^n) \\
& = I(M_2; X^n,Y^n|Z^n) \\
& = \sum_{i=1}^n H(X_{i}, Y_{i}|Z^n, X^{i-1}, Y^{i-1}) - H(X_{i}, Y_{i}|Z^n, X^{i-1}, Y^{i-1}, M_2) \\
& \ge \sum_{i=1}^n I(X_{i}, Y_{i}; U_{1i}|Z_i).
\end{align*}}
Next, consider $R_3$ {\allowdisplaybreaks
\begin{align*}
nR_3 &= H(M_3) \\
&\ge H(M_3|X^n, Y^n) \\
& \ge I(M_3; Z^n|X^n, Y^n) \\
& = H(Z^n|X^n, Y^n) - H(Z^n|X^n, Y^n, M_3)\\
& = H(Z^n|X^n, Y^n) - H(Z^n|X^n, Y^n, M_2, M_3)\\
& \ge \sum_{i=1}^n H(Z_i|X_{i}, Y_{i}) - H(Z_i|Z_{i+1}^n, X^{i}, Y^{i}, M_2, M_3) \\
& = \sum_{i=1}^n I(Z_i; U_{1i}, U_{2i}|X_i, Y_i) \\
& = \sum_{i=1}^n I(Z_i; U_{2i}|X_i, Y_i, U_{1i}),
\end{align*}}
where the last step follows from the Markov relation $Z_i - (X_i,Y_i) - U_{1i}$ which we will now prove, together with other Markov relations between the random variables. The first two Markov relations below are used for factoring the joint probability distribution while Markov relations three and four are used for establishing the distortion constraints. We will use the following lemma from~\cite{Kaspi}.
\begin{lemma}
Let $A_1, A_2, B_1, B_2$ be random variables with joint probability mass functions mf $p(a_1,a_2,b_1,b_2) = p(a_1,b_1)p(a_2,b_2).$ Let $\Mt_1$ be a function of $(A_1, A_2)$ and $\Mt_2$ be a function of $(B_1,B_2, \Mt_1).$ Then,
\begin{align}
I(A_2;B_1 |\Mt_1, \Mt_2, A_1,B_2) &= 0, \\
I(B_1;\Mt_1|A_1, B_2) &= 0, \\
I(A_2;\Mt_2|\Mt_1, A_1, B_2) & = 0.
\end{align}
\end{lemma}

Now, let us show the following Markov relations:
\begin{enumerate}
\item $Z_{i} - (X_{i}, Y_{i}) - (U_{1i}, \Xh_{1i})$: To establish this relation, we show that $I(Z_i; U_{1i}, \Xh_{1i}|X_i,Y_i) = 0.$
\begin{align*}
I(Z_i; \Xh_{1i}, U_{1i}|X_i,Y_i) &= I(Z_i; \Xh_{1i}, M_2, X^{i-1}, Y^{i-1}, Z_{i+1}^n|X_i,Y_i) \\
& \le I(Z_i; \Xh_{1i}, M_2, X^{i-1}, Y^{i-1}, X_{i+1}^n, Y_{i+1}^n, Z_{i+1}^n|X_i,Y_i) \\
& =  I(Z_i; X^{i-1}, Y^{i-1}, X_{i+1}^n, Y_{i+1}^n, Z_{i+1}^n|X_i,Y_i) \\
& = 0.
\end{align*}

\item $U_{2i} - (Z_i, U_{1i}) - (\Xh_{1i}, X_i, Y_i)$: Note that $U_{2i} = M_3.$ Consider
\begin{align*}
I(\Xh_i, X_i, Y_i ; U_{2i}|Z_i, U_{1i}) & \le I(\Xh_i, X_i^n, Y_i^n ; M_3|Z_{i}^n, X^{i-1}, Y^{i-1}, M_2) \\
& = I(X_i^n, Y_i^n ; M_3|Z_{i}^n, X^{i-1}, Y^{i-1}, M_2).
\end{align*}
Now, using Lemma 1, set $A_1 = (X^{i-1}, Y^{i-1}),$ $B_1 = Z^{i-1},$ $A_2 = (X_{i}^n, Y_{i}^n),$ $B_2 = (Z_{i}^n),$ $\Mt_2 = M_3$ and $\Mt_1 = M_2.$ Then, using the third expression in the Lemma, we see that $I(X_i^n, Y_i^n ; M_3|Z_{i}^n, X^{i-1}, Y^{i-1}, M_2)=0.$

\item $Z^{i-1} - (U_{1i}, Z_i) - (X_i, Y_i)$: Consider{\allowdisplaybreaks
\begin{align*}
 I(X_i,Y_i; Z^{i-1}|U_{1i}, Z_i) & \le I(X_{i}^n, Y_{i}^n; Z^{i-1}| X^{i-1}, Y^{i-1}, Z_{i}^{n}, M_2) \\
& = H(Z^{i-1}|X^{i-1}, Y^{i-1}, Z_{i}^{n}, M_2) - H(Z^{i-1}|X^n, Y^n, Z_{i}^{n}, M_2) \\
& \le H(Z^{i-1}|X^{i-1}, Y^{i-1}, Z_{i}^{n}) - H(Z^{i-1}|X^n, Y^n, Z_{i}^{n}) \\
& = H(Z^{i-1}|X^{i-1}, Y^{i-1}) - H(Z^{i-1}|X^{i-1}, Y^{i-1}) \\
& = 0.
\end{align*}}

\item $(X_{i+1}^n, Y_{i+1}^n) - (U_{1i}, U_{2i}, X_{i}, Y_{i}) - Z_i$: Consider
\begin{align*}
I(X_{i+1}^n, Y_{i+1}^n; Z_i|U_{1i}, U_{2i}, X_{i}, Y_{i}) \le I(X_{i+1}^n, Y_{i+1}^n; Z^i|M_2,M_3, Z_{i+1}^n, X^{i}, Y^{i}).
\end{align*}
Applying the first expression in the Lemma with $A_2 = (X_{i+1}^n, Y_{i+1}^n),$ $A_1 = (X^i, Y^i),$ $B_1 = Z^i$ and $B_2 = Z_{i+1}^n$ gives $I(X_{i+1}^n, Y_{i+1}^n; Z_i|U_{1i}, U_{2i}, X_{i}, Y_{i}) = 0.$ 
\end{enumerate}

\noindent \textit{Distortion constraints}

We show that the auxiliary definitions satisfy the distortion constraints by showing the existence of functions $\xh_{2i}^*(U_{1i}, Z_i)$ and $\zh_{i}^*(U_{1i}, U_{2i}, X_i, Y_i)$ such that 
\begin{align*}
\E(d_2(X_{i}, \xh_{2i}^*(U_{1i}, Z_i))) & \le \E(d_2(X_i, \xh_{2i}(M_2, Z^n))) \\
\E(d_2(Z_{i}, \zh_{i}^*(U_{1i}, U_{2i}, X_i, Y_i))) & \le \E(d_2(X_i, \zh_{3i}(M_3, X^n,Y^n, Z^n))),
\end{align*}
where $\xh_{2i}(M_2, Z^n)$ and $\zh_{i}(M_3, X^n,Y^n)$ are the original reconstruction functions.

To prove the first expression, we have
\begin{align*}
\E(d_2(X_i, \xh_{2i}(M_2, Z^n))) &= \sum p(x^{i}, y^{i}, z^n, m_2)d_2(x_i, \xh_{2i}(m_2, z^n)) \\
& \stackrel{(a)}{=} \sum p(u_{1i}, z^i)p(x_i,y_i|u_{1i}, z^i)d_2(x_i, \xh'_{2i}(u_{1i}, z_i, z^{i-1})) \\
& = \sum p(u_{1i}, z_i, z^{i-1})p(x_i,y_i|u_{1i}, z_i)d_2(x_i, \xh'_{2i}(u_{1i}, z_i, z^{i-1})),
\end{align*}
where (a) follows from defining $\xh'_{2i}(u_{1i}, z_i, z^{i-1}) = \xh_{2i}(m_2, z^n)$ for all $x^{i-1}, y^{i-1}$ and the last step follows from the Markov relation $Z^{i-1} - (U_{1i}, Z_i) - (X_i, Y_i).$ Finally, defining \\ $(z^{i-1})^* = \arg \min_{z^{i-1}}\sum_{x_i,y_i}p(x_i,y_i|u_{1i}, z_i)d_2(x_i, \xh'_{2i}(u_{1i}, z_i, z^{i-1}))$ and $\xh_{2i}^*(u_{1i}, z_i) = \xh'_{2i}(u_{1i}, z_i, (z^{i-1})^*)$ shows that $\E(d_2(X_{i}, \xh_{2i}^*(U_{1i}, Z_i))) \le \E(d_2(X_i, \xh_{2i}(M_2, Z^n)))$ as required. \\

To prove the second expression, we follow similar steps. Considering the expected distortion, we have{\allowdisplaybreaks 
\begin{align*}
& \E(d_3(Z_i, \zh_{i}(M_3, X^n, Y^n))) \\
&= \sum p(z_i^n, x^n, y^n, m_3)d_3(z_i, \zh_{i}(m_3, x^n,y^n)) \\
&= \sum p(u_{1i}, u_{2i}, x_i, y_i, x_{i+1}^n, y_{i+1}^n)p(z_i| u_{1i}, u_{2i}, x_i, y_i, x_{i+1}^n, y_{i+1}^n)d_3(z_i, \zh'_{3i}(u_{1i},u_{2i}, x_i,y_i, x_{i+1}^n, y_{i+1}^n)) \\
& = \sum p(u_{1i}, u_{2i}, x_i, y_i, x_{i+1}^n, y_{i+1}^n)p(z_i| u_{1i}, u_{2i}, x_i, y_i)d_3(z_i, \zh'_{i}(u_{1i},u_{2i}, x_i,y_i, x_{i+1}^n, y_{i+1}^n)), 
\end{align*}}
where the last step uses Markov relation 4. The rest of the proof is omitted since it uses the same steps as the proof for the first distortion constraint. 

Finally, using the standard time sharing random variable $Q$ as before and defining $U_{1} = (U_{1Q}, Q),$ $U_{2} = U_{2Q},$ $\Xh_{1} = \Xh_{1Q}$, we obtain the required outer bound for the rate-distortion region.
\end{proof}
We now turn to the final case of Two Way Triangular Source Coding.
\subsection{Two Way Triangular Source Coding}
\begin{theorem}[Rate Distortion Region for Two Way Triangular Source Coding] \label{thm4}

$\Rc(D_1,D_2, D_3)$ for Two Way Triangular Source Coding is given by the set of all rate tuples $(R_1, R_2, R_3, R_4)$ satisfying
\begin{align}
R_1 &\ge I(X;\Xh_1, U_1|Y), \label{trs}\\ 
R_2 &\ge I(X, Y; U_1|Z), \\ 
R_3 &\ge I(X, Y; V|Z, U_1), \\
R_4 &\ge I(U_2; Z|U_1, V, X, Y), \label{trf}
\end{align}
for some $p(x,y,z, u_1, u_2, v, \xh_1) = p(x)p(y|x)p(z|y)p(u_1|x,y)p(\xh_1|x,y,u_1)p(v|x,y,u_1)p(u_2|z,u_1,v)$ and functions $g_2 : \Uc_1 \times \Vc\times \Zc \rightarrow \Xh_2$ and $g_3 : \Uc_1 \times \Uc_2 \times \Vc\times\Xc \times \Yc \rightarrow \Zh$ such that
\begin{align}
\E(d_1(X, \Xh_1)) &\le D_1, \label{td1} \\ 
\E(d_2(X, \Xh_2)) &\le D_2, \label{td2} \\ 
\E(d_3(Z, \Zh)) &\le D_3. \label{td3}
\end{align}
The cardinalities for the auxiliary random variables are upper bounded by $|\Uc_1| \le |\Xc||\Yc|+6$, $|\Vc| \le |\Uc_1|(|\Xc||\Yc|+3)$ and $|\Uc_2| \le |\Uc_1||\Vc|(|\Zc|+1)$.
\end{theorem}

\noindent \textit{Sketch of Achievability}

The forward direction ($R_1$, $R_2$, $R_3$) for Two-Way triangular source coding follows the procedure in Theorem~\ref{thm2}. For the reverse direction ($R_4$), it follows Theorem~\ref{thm3} with $(U_1,V)$ replacing the role of $U_1$ in Theorem~\ref{thm3}.

\begin{proof}[Proof of Converse]
Given a $(n, 2^{nR_1}, 2^{nR_2}, 2^{nR_3}, 2^{nR_4}, D_1, D_2, D_3)$ code, define $U_{1i} = (M_2, X^{i-1}, Y^{i-1}, Z_{i+1}^n)$, $U_{2i} =M_4$ and $V_i = (M_3, U_{1i})$. The $R_1$ and $R_2$ bounds follow the same steps as in Theorem 3. For $R_3$, we have
\begin{align*}
nR_3 &\ge H(M_3) \\
&\ge H(M_3|M_2,Z^n) \\
& = I(X^n,Y^n;M_3|M_2, Z^n) \\
& = \sum_{i=1}^n H(X_i,Y_i|M_2, Z^n, X^{i-1}, Y^{i-1}) - H(X_i,Y_i|M_2, M_3, Z^n, X^{i-1}, Y^{i-1}) \\
& {\ge} \sum_{i=1}^n H(X_i,Y_i|U_i, Z_i) - H(X_i,Y_i|U_{1i},V_i, Z_i) \\
& = \sum_{i=1}^n I(X_i,Y_i; V_i|U_{1i},Z_i).
\end{align*}
Next, consider{\allowdisplaybreaks
\begin{align*}
nR_4 &= H(M_4) \\
&\ge H(M_4|X^n, Y^n) \\
& \ge I(M_4; Z^n|X^n, Y^n) \\
& = H(Z^n|X^n, Y^n) - H(Z^n|X^n, Y^n, M_4)\\
& = H(Z^n|X^n, Y^n) - H(Z^n|X^n, Y^n, M_2, M_3, M_4)\\
& \ge \sum_{i=1}^n H(Z_i|X_{i}, Y_{i}) - H(Z_i|Z_{i+1}^n, X^{i}, Y^{i}, M_2, M_3, M_4) \\
& = \sum_{i=1}^n I(Z_i; U_{1i}, V_{i}, U_{2i}|X_i, Y_i)\\
& = \sum_{i=1}^n I(Z_i; U_{2i}|X_i, Y_i, V_{i}, U_{1i}),
\end{align*}}
where the last step follows from the Markov relation $Z_i - (X_i, Y_i) - (V_i, U_{1i})$ which we will now prove together with other Markov relations between the random variables. The first 2 Markov relations are for factoring the probability distribution while Markov relations 3 and 4 are for establishing the distortion constraints.

\noindent \textit{Markov Relations}
\begin{enumerate}
\item $Z_{i} - (X_{i}, Y_{i}) - (U_{1i}, V_i, \Xh_{1i})$: To establish this relation, we show that $I(Z_i; \Xh_{1i}, U_{1i}, V_i|X_i,Y_i) = 0.$
\begin{align*}
I(Z_i; \Xh_{1i}, U_{1i}, V_i|X_i,Y_i) &= I(Z_i; \Xh_{1i}, M_3, M_2, X^{i-1}, Y^{i-1}, Z_{i+1}^n|X_i,Y_i) \\
& \le I(Z_i; \Xh_{1i}, M_3, M_2, X^{i-1}, Y^{i-1}, X_{i+1}^n, Y_{i+1}^n, Z_{i+1}^n|X_i,Y_i) \\
& =  I(Z_i; X^{i-1}, Y^{i-1}, X_{i+1}^n, Y_{i+1}^n, Z_{i+1}^n|X_i,Y_i) \\
& = 0.
\end{align*}

\item $U_{2i} - (Z_i, U_{1i}, V_i) - (\Xh_{1i}, X_i, Y_i)$: Consider
\begin{align*}
I(\Xh_i, X_i, Y_i ; U_{2i}|Z_i, U_{1i}, V_i) & \le I(\Xh_i, X_i^n, Y_i^n ; M_4|Z_{i}^n, X^{i-1}, Y^{i-1}, M_2, M_3) \\
& = I(X_i^n, Y_i^n ; M_4|Z_{i}^n, X^{i-1}, Y^{i-1}, M_2, M_3).
\end{align*}
Now, using Lemma 1, set $A_1 = (X^{i-1}, Y^{i-1}),$ $B_1 = Z^{i-1},$ $A_2 = (X_{i}^n, Y_{i}^n),$ $B_2 = (Z_{i}^n),$ $\Mt_2 = M_4$ and $\Mt_1 = M_2.$ Then, using the third expression in the Lemma, we see that $I(X_i^n, Y_i^n ; M_4|Z_{i}^n, X^{i-1}, Y^{i-1}, M_2)=0.$

\item $Z^{i-1} - (U_{1i}, V_i, Z_i) - (X_i, Y_i)$: Consider{\allowdisplaybreaks
\begin{align*}
 I(X_i,Y_i; Z^{i-1}|U_{1i}, V_i, Z_i)  &\le I(X_{i}^n, Y_{i}^n; Z^{i-1}| X^{i-1}, Y^{i-1}, Z_{i}^{n}, M_2,M_3) \\
& = H(Z^{i-1}|X^{i-1}, Y^{i-1}, Z_{i}^{n}, M_2,M_3) - H(Z^{i-1}|X^n, Y^n, Z_{i}^{n}, M_2,M_3) \\
& \le H(Z^{i-1}|X^{i-1}, Y^{i-1}, Z_{i}^{n}) - H(Z^{i-1}|X^n, Y^n, Z_{i}^{n}) \\
& = H(Z^{i-1}|X^{i-1}, Y^{i-1}) - H(Z^{i-1}|X^{i-1}, Y^{i-1}) \\
& = 0.
\end{align*}}
\item $(X_{i+1}^n, Y_{i+1}^n) - (U_{1i}, U_{2i}, V_i, X_{i}, Y_{i}) - Z_i$: Consider
\begin{align*}
I(X_{i+1}^n, Y_{i+1}^n; Z_i|U_{1i}, U_{2i}, V_i, X_{i}, Y_{i}) \le I(X_{i+1}^n, Y_{i+1}^n; Z^i|M_2,M_3, M_4, Z_{i+1}^n, X^{i}, Y^{i}).
\end{align*}
Applying the first expression in the Lemma with $A_2 = (X_{i+1}^n, Y_{i+1}^n),$ $A_1 = (X^i, Y^i),$ $B_1 = Z^i$ and $B_2 = Z_{i+1}^n$ gives $I(X_{i+1}^n, Y_{i+1}^n; Z_i|U_{1i}, U_{2i}, X_{i}, Y_{i}) = 0.$ 
\end{enumerate}

\noindent \textit{Distortion Constraints}

The proof of the distortion constraints is omitted since it follows similar steps to the Two Way Cascade Source Coding case, with the new Markov relations 3 and 4, and $(U_{1i}, V_i)$ replacing $U_{1i}$ in the proof.

Using the standard time sharing random variable $Q$ as before and defining $U_{1} = (U_{1Q}, Q),$ $U_{2} = U_{2Q},$ $\Xh_{1} = \Xh_{1Q}$ and $V = V_Q$ we obtain an outer bound for the rate-distortion region for some probability distribution of the form $p(x,y,z, u_1, u_2, v, \xh_1) = p(x,y,z)p(u_1|x,y)p(\xh_1, v|x,y,u_1)p(u_2|z,u_1,v)$. It remains to show that it suffices to consider probability distributions of the form $p(x,y,z)p(u_1|x,y)p(\xh_1|x,y,u_1)p(v|x,y,u_1)p(u_2|z,u_1,v)$. This follows similar steps to proof of Theorem 2. Let
\begin{align*}
p_1 &= p(x,y,z)p(u_1|x,y)p(\xh_1, v|x,y,u_1)p(u_2|z,u_1,v), \\
p_2 &= p(x,y,z)p(u_1|x,y)\hat{p}(\xh_1|x,y,u_1)\hat{p}(v|x,y,u_1)p(u_2|z,u_1,v),
\end{align*}
where $\hat{p}(\xh_1|x,y,u_1)$ and $\hat{p}(v|x,y,u_1)$ are the marginals induced by $p_1$. Next, note that $R_1$, $R_2$, $R_3$, $R_4$ and the distortion constraints depend on $p_1$ only through the marginals $p(x,y,z,u_1,u_2,v)$ and $p(x,y,z,u_1, \xh_1)$. Since these marginals are the same for $p_1$ and $p_2$, the rate and distortion constraints are unchanged. Finally, note that the Markov relations 1 and 2 implied by $p_1$ continues to hold under $p_2$. This completes the proof of the converse.
\end{proof}
\section{Gaussian Quadratic Distortion Case} \label{sect:4}
In this section, we evaluate the rate-distortion regions when $(X,Y,Z)$ are jointly Gaussian and the distortion is measured in terms of the mean square error. We will assume, without loss of generality, that $X = A + B+Z$, $Y = B+Z$ and $Z = Z$, where $A$, $B$ and $Z$ are independent, zero mean Gaussian random variables with variances $\sigma_A^2$, $\sigma_B^2$ and $\sigma_Z^2$ respectively.
\subsection{Quadratic Gaussian Cascade Source Coding}
\begin{corollary}[Quadratic Gaussian Cascade Source Coding] \label{coro1}
First, we note that if $R_2 < \frac{1}{2}\log \frac{\sigma_{A}^2 + \sigma_{B}^2}{D_2}$, then the distortion constraint $D_2$ cannot be met. Hence, given $D_1, D_2>0$ and $R_2 \ge \max\{\frac{1}{2}\log \frac{\sigma_{A}^2 + \sigma_{B}^2}{D_2},0\}$, the rate distortion region for Quadratic Gaussian Cascade Source Coding is characterized by the smallest rate $R_1$ such that $(D_1, D_2, R_1, R_2)$ are achievable, which is 
\begin{align*}
R_1 = \max\left\{\frac{1}{2}\log \frac{\sigma_{A}^2}{D_1}, \frac{1}{2}\log \frac{\sigma_{A}^2}{\sigma_{A|U,B}^2}\right\},
\end{align*}
where $U = \alpha^*A + \beta^* B + Z^*$, $Z^* \sim N (0, \sigma_{Z^*}^2)$, with $\alpha^*$, $\beta^*$ and $\sigma_{Z^*}^2$ achieving the maximum in the following optimization problem:
\begin{align*}
\mbox{maximize} & \quad \sigma_{A|U,B}^2 \\
\mbox{subject to} & \quad R_2 \ge \frac{1}{2}\log \frac{\sigma_U^2}{\sigma_{Z^*}^2} \\
& \quad D_2 \ge \sigma_{A+B|U}^2
\end{align*}
\end{corollary}
The optimization problem given in the corollary can be solved following analysis in~\cite{Permuter}. In our proof of the corollary, we will show that the rate distortion region obtained is the same as the case when the degraded side information $Z$ is available to all nodes. 

\begin{proof}[Converse]
Consider the case when the side information $Z$ is available to all nodes. Without loss of generality, we can subtract the side information away from $X$ and $Y$ to obtain a rate distortion problem involving only $A+B$ and $B$ at Node 0, $B$ at Node 1 and no side information at Node 2. Characterization of this class of Quadratic Gaussian Cascade source coding problem has been carried out in~\cite{Permuter} and following the analysis therein, we can show that the rate distortion region is given by the region in Corollary~\ref{coro1}. 
\end{proof}
\begin{proof}[Achievability]
We evaluate Theorem~\ref{thm1} using Gaussian auxiliaries random variables. Let $U' = \alpha^*X + (\beta^* -\alpha^*)Y + Z^*  = \alpha^*A + \beta^* (B+Z) + Z^*$ and $V$ be a Gaussian random variable that we will specify in the proof. We now rewrite $R_1 = I(X;U',\Xh_1|Y)$ as $R_1 = I(X;U',V|Y)$ with $\Xh_1 = V + \E(X|U',Y)$, $V$ independent of $U'$ and $Y$. Let $g_2(U', Z) = \E(X|U',Z)$. Evaluating $R_1$ and $R_2$ using this choice of auxiliaries, we have
\begin{align*}
R_1 &= I(X;U',V|Y) \\
& = h(A+B+Z|B+Z) - h(X|U',V,Y) \\
& = \frac{1}{2}\log\frac{\sigma_{A}^2}{\sigma_{X|U',V,Y}^2}, \\
R_2 & = I(X,Y;U'|Z) \\
& = h(U'|Z) - h(U'|X,Y,Z) \\
& = \frac{1}{2}\log \frac{\sigma_{\alpha^*A+ \beta^*B + Z^*}^2}{\sigma_{Z^*}^2} \\
& = \frac{1}{2}\log \frac{\sigma_{U}^2}{\sigma_{Z^*}^2}.
\end{align*}
Next, we have{\allowdisplaybreaks
\begin{align*}
\sigma_{X|U', Y}^2 &= \sigma_{A+B+Z|\alpha^*A + \beta^* (B+Z) + Z^*, B+Z}^2 \\
& = \sigma_{A|\alpha^*A + Z^*, B+Z}^2 \\
& = \sigma_{A|\alpha^*A + Z^*}^2 \\
& = \sigma_{A|U,B}^2 .
\end{align*}}
If $\sigma_{X|U', Y}^2 = \sigma_{A|U,B}^2 \le D_1$, we set $V = 0$ to obtain $R_1 = \frac{1}{2}\log\frac{\sigma_{A}^2}{\sigma_{X|U',Y}^2}$. If $\sigma_{X|U', Y}^2 > D_1$, then we choose $V = X - \E(X|U',Y) + Z_2$ where $Z_2\sim N(0, D_1\sigma_{X|U', Y}^2/(\sigma_{X|U', Y}^2 - D_1))$ so that $\sigma_{X|U',V,Y}^2 = D_1$ and obtain $R_1 = \frac{1}{2}\log\frac{\sigma_{A}^2}{D_1}$. Therefore, $R_1 = \max\{\frac{1}{2}\log \frac{\sigma_{A}^2}{D_1}, \frac{1}{2}\log \frac{\sigma_{A}^2}{\sigma_{A|U,B}^2}\}$. 

Finally, we show that this choice of random variables satisfy the distortion constraints. For $D_1$, note that since $\E (X - \Xh_1)^2 = \sigma_{X|U',V,Y}^2$, the distortion constraint $D_1$ is always satisfied. For the second distortion constraint, we have 
\begin{align*}
\E(X-\Xh_2)^2 & = \sigma_{X|U', Z}^2 \\
&= \sigma_{A+B|\alpha^*A + \beta^* (B+Z) + Z^*, Z}^2 \\
& = \sigma_{A+B|\alpha^*A+\beta^* B + Z^*, Z}^2 \\
& =  \sigma_{A+B|\alpha^*A+\beta^* B + Z^*}^2 \\
& = \sigma_{A+B|U}^2 \\
&\le D_2. 
\end{align*}

Hence, our choice of auxiliary $U'$ and $V$ satisfies the rate distortion region and distortion constraints given in the corollary, which completes our proof.
\end{proof}

\subsection{Quadratic Gaussian Triangular Source Coding}
\begin{corollary}[Quadratic Gaussian Triangular Source Coding] \label{coro2}
Given $D_1, D_2>0$ and $R_2, R_3 \ge 0, \, R_2 + R_3 \ge \frac{1}{2}\log \frac{\sigma_{A}^2 + \sigma_{B}^2}{D_2}$, the rate distortion region for Quadratic Gaussian Triangular Source Coding is characterized by the smallest $R_1$ for which $(D_1, D_2, R_1, R_2, R_3)$ is achievable, which is 
\begin{align*}
R_1 = \max\left\{\frac{1}{2}\log \frac{\sigma_{A}^2}{D_1}, \frac{1}{2}\log \frac{\sigma_{A}^2}{\sigma_{A|U,B}^2}\right\},
\end{align*}
where $U = \alpha^*A + \beta^* B + Z^*$, $Z \sim N (0, \sigma_{Z^*}^2)$, with $\alpha^*$, $\beta^*$ and $\sigma_{Z^*}^2$ satisfying the following optimization problem.
\begin{align*}
\mbox{maximize} & \quad \sigma_{A|U,B}^2 \\
\mbox{subject to} & \quad R_2 \ge \frac{1}{2}\log \frac{\sigma_U^2}{\sigma_{Z^*}^2} \\
& \quad 2^{2R_3}D_2 \ge \sigma_{A+B|U}^2
\end{align*}
\end{corollary}

As with Corollary~\ref{coro1}, the optimization problem given this corollary can be solved following analysis in~\cite{Permuter}.

\begin{proof}[Converse]
The converse uses the same approach as Corollary~\ref{coro1}. Consider the case when the side information $Z$ is available to all nodes. Without loss of generality, we can subtract the side information away from $X$ and $Y$ to obtain a rate distortion problem involving only $A+B$ and $B$ at Node 0, $B$ at Node 1 and no side information at Node 2. Characterization of this class of Quadratic Gaussian Triangular source coding problem has been carried out in~\cite{Permuter} and following the analysis therein, we can show that the rate distortion region is given by the region in Corollary~\ref{coro2}. 
\end{proof}
\begin{proof}[Achievability]
We evaluate Theorem~\ref{thm2} using Gaussian auxiliary random variables. Let $U' = \alpha^*X + (\beta^* -\alpha^*)Y + Z^*  = \alpha^*A + \beta^* (B+Z) + Z^*$ and $V' = X + \eta U' + Z_3$, $Z_3\sim N(0, \sigma_{Z_3}^2)$. Following the analysis in Corollary~\ref{coro1}, the inequalities for the rates are{\allowdisplaybreaks 
\begin{align*}
R_1 &= \max\left\{\frac{1}{2}\log \frac{\sigma_{A}^2}{D_1}, \frac{1}{2}\log \frac{\sigma_{A}^2}{\sigma_{A|U,B}^2}\right\}, \\
R_2 &\ge \frac{1}{2}\log \frac{\sigma_U^2}{\sigma_{Z^*}^2}, \\
R_3 &\ge I(X,Y;V|Z,U') = I(X;V'|Z,U') \\
&=\frac{1}{2}\log \frac{\sigma_{X|Z,U'}^2}{\sigma_{X|Z,U', V'}^2}.
\end{align*}}

As with Corollary~\ref{coro1}, the distortion constraint $D_1$ is satisfied with an appropriate choice of $\Xh_1$. For the distortion constraint $D_2$, we have
\begin{align*}
D_2 \ge \sigma_{X|Z,U', V'}^2.
\end{align*}

Next, note that we can assume equality for $R_3$, since we can adjust $\eta$ and $\sigma_{Z_3}^2$ so that inequality is met. Since this operation can will only decrease $\sigma_{X|Z,U', V'}^2$, the distortion constraint $D_2$ will still be met. Therefore, setting $R_3 = \frac{1}{2}\log \frac{\sigma_{X|Z,U'}^2}{\sigma_{X|Z,U', V'}^2}$, we have
\begin{align*}
D_2 &\ge \sigma_{X|Z,U', V'}^2 \\
&= \frac{\sigma_{X|Z,U'}^2}{2^{2R_3}}.
\end{align*}
Since $\sigma_{X|Z,U'}^2 = \sigma_{A+B|U}^2$, this completes the proof of achievability.
\end{proof}

\noindent {\em Remark: } As alternative characterizations, we show in Appendix~\ref{appen:3} that the Cascade and Triangular settings in Corollaries~\ref{coro1} and~\ref{coro2} can be transformed into equivalent problems in~\cite{Permuter} where explicit characterizations of the rate distortion regions were given.

\subsection{Quadratic Gaussian Two Way Source Coding}
It is straightforward to extend Corollaries~\ref{coro1} and~\ref{coro2} to Quadratic Gaussian Two Way Cascade and Triangular Source Coding using the observation that in the Quadratic Gaussian case, side information at the encoder does not reduce the required rate. Therefore, the backward rate from Node 2 to Node 0 is always lower bounded by $\frac{1}{2}\log\frac{\sigma_{Z|B+Z}^2}{D_3}$. This rate (and distortion constraint $D_3$) can be achieved by simply encoding $Z$. We therefore state the following corollary without proof.

\begin{corollary}[Quadratic Gaussian Two Way Triangular Source Coding] \label{coro4}
Given $D_1, D_2, D_3>0$, $R_2, R_3 \ge 0, \, R_2 + R_3 \ge \frac{1}{2}\log \frac{\sigma_{A}^2 + \sigma_{B}^2}{D_2}$ and $R_4 \ge \max\{\frac{1}{2}\log \frac{\sigma_{Z|Y}^2}{D_3}, 0\} $, the rate distortion region for Quadratic Gaussian Two Way Triangular Source Coding is characterized by the smallest $R_1$ for which $(R_1, R_2, R_3, R_4, D_1, D_2, D_3)$ is achievable, which is
\begin{align*}
R_1 = \max\left\{\frac{1}{2}\log \frac{\sigma_{A}^2}{D_1}, \frac{1}{2}\log \frac{\sigma_{A}^2}{\sigma_{A|U,B}^2}\right\},
\end{align*}
where $U = \alpha^*A + \beta^* B + Z^*$, $Z \sim N (0, \sigma_{Z^*}^2)$, with $\alpha^*$, $\beta^*$ and $\sigma_{Z^*}^2$ satisfying the following optimization problem.
\begin{align*}
\mbox{maximize} & \quad \sigma_{A|U,B}^2 \\
\mbox{subject to} & \quad R_2 \ge \frac{1}{2}\log \frac{\sigma_U^2}{\sigma_{Z^*}^2} \\
& \quad 2^{2R_3}D_2 \ge \sigma_{A+B|U}^2
\end{align*}
\end{corollary}

\textit{Remark:} The special case of Two Way Cascade Quadratic Gaussian Source Coding can be obtained as a special case by setting $R_3 =0$.

Next, we present an extension to our settings for which we can characterize the rate-distortion region in the Quadratic Gaussian case. In this extended setting, we have Cascade setting from Node 0 to Node 2 and a triangular setting from Node 2 to Node 0, with the additional constraint that Node 1 also reconstructs a lossy version of $Z$. As formal definitions are natural extensions of those presented in section~\ref{sect:2}, we will omit them here. The setting is shown in Figure~\ref{fig4}.

\begin{figure} [!h] 
\begin{center}
\psfrag{x}[l]{\large $X,Y$}
\psfrag{y}[c]{\large $Y$}
\psfrag{z}[c]{\large $Z$}
\psfrag{r1}{$R_1$}
\psfrag{r2}{$R_2$}
\psfrag{r4}{$R_3$}
\psfrag{r5}{$R_4$}
\psfrag{r6}{$R_5$}
\psfrag{x1}[c]{$\Xh_1, \Zh_1$}
\psfrag{x2}{$\Xh_2$}
\psfrag{z2}[c]{$\Zh_2$}
\psfrag{n0}[l]{Node 0}
\psfrag{n1}[l]{Node 1}
\psfrag{n2}[l]{Node 2}
\scalebox{0.85}{\includegraphics{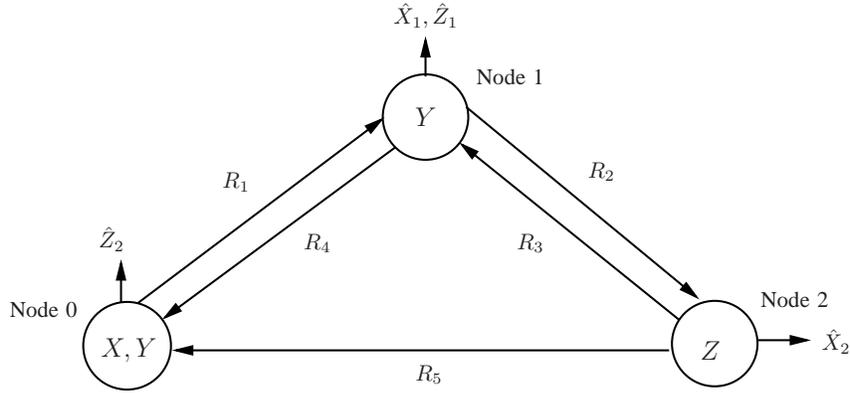}} 
\end{center}
\caption{Extended Quadratic Gaussian Two Way source coding} \label{fig4}
\end{figure}
\begin{theorem}[Extended Quadratic Gaussian Two Way Cascade Source Coding] \label{coro5}
Given $D_1, D_2>0$, $0< D_{Z_1}, D_{Z_2} \le \sigma_{Z|Y}^2$ and $R_2 \ge \max\{\frac{1}{2}\log \frac{\sigma_{A}^2 + \sigma_{B}^2}{D_2},0\}$, the rate distortion region for the Extended Quadratic Gaussian Two Way Cascade Source Coding is given by the set of $R_1, R_3, R_4, R_5 \ge 0$ satisfying the following equalities and inequalities
\begin{align*}
R_1 = \max\left\{\frac{1}{2}\log \frac{\sigma_{A}^2}{D_1}, \frac{1}{2}\log \frac{\sigma_{A}^2}{\sigma_{A|U,B}^2}\right\},
\end{align*}
where $U = \alpha^*A + \beta^* B + Z^*$, $Z^* \sim N (0, \sigma_{Z^*}^2)$, with $\alpha^*$, $\beta^*$ and $\sigma_{Z^*}^2$ satisfying the following optimization problem.
\begin{align*}
\mbox{maximize} & \quad \sigma_{A|U,B}^2 \\
\mbox{subject to} & \quad R_2 \ge \frac{1}{2}\log \frac{\sigma_U^2}{\sigma_{Z^*}^2} \\
& \quad D_2 \ge \sigma_{A+B|U}^2
\end{align*}
and{\allowdisplaybreaks
\begin{align*}
R_3 &\ge \frac{1}{2}\log\frac{\sigma_{Z|Y}^2}{D_{Z_1}}, \\
R_3 + R_5 &\ge \frac{1}{2}\log\frac{\sigma_{Z|Y}^2}{\min\{D_{Z_1}, D_{Z_2}\}}, \\
R_4 + R_5 &\ge \frac{1}{2}\log\frac{\sigma_{Z|Y}^2}{D_{Z_2}}.
\end{align*}}
\end{theorem}

\begin{proof}

\textit{Converse}

For the forward direction $(R_1, R_2)$, we note that Node 2 can only send a function of $(M_1, Y^n, Z^n)$ to Nodes 0 and 1 using the $R_4$ and $R_5$ links. Since $M_1$ and $Y^n$ available at both Node 0 and 1, the forward rates are lower bounded by the setting where $Z^n$ is available to all nodes. Further, in this setting, the distortion constraints $D_{Z_1}$ and $D_{Z_2}$ are automatically satisfied since $Z$ is available at Nodes 0 and 1. Therefore, $(R_3, R_4, R_5)$ do not affect the achievable $(R_1, R_2)$ rates in this modified (lower bound) setting. $(R_1, R_2)$ are then obtained by the observation in Corollary~\ref{coro1} that the rate distortion region obtained for our Quadratic Gaussian Cascade setting in Corollary~\ref{coro1} is equivalent to the case where the side information $Z$ is available at all nodes.

For the reverse direction, the lower bounds are derived by letting the side information $(X,Y)$ to be available at Node 2, and for side information $X$ to be available at Node 1. The $D_1$ and $D_2$ distortion constraints are then automatically satisfied since $X$ is available at all nodes. We then observed that $(R_1, R_2)$ do not affect the achievable $(R_3, R_4, R_5)$ rates in this modified (lower bound) setting. The stated inequalities for $R_3, R_4, R_5$ are then obtained from standard cutset bound arguments and the fact that $X-Y-Z$ form a Markov Chain. 

\textit{Achievability}

We analyze only the backward rates $R_3, R_4$ and $R_5$ since the forward direction follows from Corollary~\ref{coro1}. For the backward rates, we now show that the rates are achievable without the assumption of $(X,Y)$ being available at Node 2. We will rely on results on successive refinement of Gaussian sources with common side information given in~\cite{Steinberg--Merhav}. A simplified figure of the setup for analyzing the backward rates is given in Figure~\ref{fig5}. We have three cases to consider.

\begin{figure} [!h] 
\begin{center}
\psfrag{x}[c]{\large $Y$}
\psfrag{y}[c]{\large $Y$}
\psfrag{z}[c]{\large $Z$}
\psfrag{r4}{$R_3$}
\psfrag{r5}{$R_4$}
\psfrag{r6}{$R_5$}
\psfrag{x1}[C]{$\Zh_1$}
\psfrag{z2}{$\Zh_2$}
\psfrag{n0}[l]{Node 0}
\psfrag{n1}[l]{Node 1}
\psfrag{n2}[l]{Node 2}
\scalebox{0.75}{\includegraphics{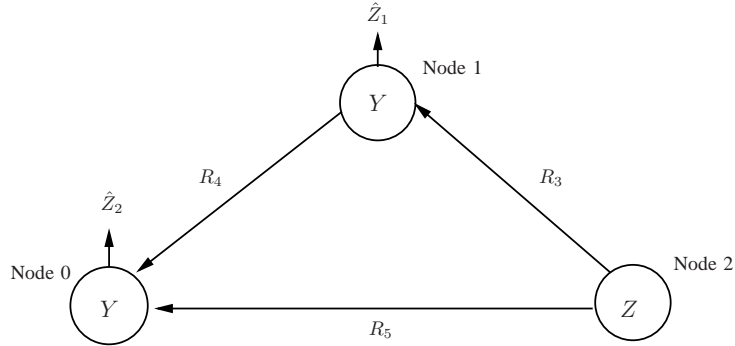}} 
\end{center}
\caption{Setup for analysis of achievability of backward rates} \label{fig5}
\end{figure} 

{\em Case 1: $D_{Z_1} \le D_{Z_2}$}

In this case, the inequalities in the lower bound reduce to 
\begin{align*}
R_3 &\ge \frac{1}{2}\log\frac{\sigma_{Z|Y}^2}{D_{Z_1}}, \\
R_4 + R_5 &\ge \frac{1}{2}\log\frac{\sigma_{Z|Y}^2}{D_{Z_2}}.
\end{align*}

From the successive refinement results in~\cite{Steinberg--Merhav}, we can show that the following rates are achievable
\begin{align*}
R_3 &= I(U_1, U_2, U_3;Z|Y), \\
R_4 &= I(U_2;Z|Y), \\
R_5 &= I(U_3;Z|Y, U_2)
\end{align*}
for some conditional distribution $F(U_1, U_2, U_3|Z)$, $\Zh_1(U_1, U_2, U_3, Y)$ and $\Zh_2(U_1, U_2,Y)$ satisfying the distortion constraints. Now, for fixed $R_4 \le \frac{1}{2}\log\frac{\sigma_{Z|Y}^2}{D_{Z_2}}$, choose $D' (\ge D_{Z_2})$ such that $R_4 = \frac{1}{2}\log\frac{\sigma_{Z|Y}^2}{D'}$. We now choose the auxiliary random variables and reconstruction functions in the following manner. Define $Q(x):= \frac{x \, \sigma_{Z|Y}^2}{\sigma_{Z|Y}^2 - x}$.
\begin{align*}
U_1 &= Z + W_1 \mbox{ where } W_1 \sim N(0, Q(D_{Z_1})),\\
U_3 &= U_1 + W_3 \mbox{ where } W_3 \sim N(0, Q(D_{Z_2})-Q(D_{Z_1})),\\
U_2 &= U_3 + W_2 \mbox{ where } W_2 \sim N(0, Q(D')-Q(D_{Z_2})),\\
\Zh_1 &= \E(Z|U_1, Y), \\
\Zh_2 &= \E(Z|U_3, Y).
\end{align*}

From this choice of auxiliary random variables, it is easy to verify the following{\allowdisplaybreaks
\begin{align*}
R_3 & = I(U_1, U_2, U_3;Z|Y) \\
& = I(U_1;Z|Y) \\
&= \frac{1}{2}\log\frac{\sigma_{Z|Y}^2}{D_{Z_1}}, \\
R_4 &= I(U_2;Z|Y) \\
& = \frac{1}{2}\log\frac{\sigma_{Z|Y}^2}{D'}, \\
R_4 + R_5 & = I(U_2;Z|Y) + I(U_3;Z|Y, U_2) \\
& = I(U_3,U_2;Z|Y) \\
&= \frac{1}{2}\log\frac{\sigma_{Z|Y}^2}{D_{Z_2}}, \\
\E(Z - \Zh_1)^2 &= D_{Z_1}, \\
\E(Z - \Zh_2)^2 &= D_{Z_2}.
\end{align*} }

{\em Case 2: $D_{Z_1} > D_{Z_2}$, $R_3 \ge R_4$}

In this case, the active inequalities are
\begin{align*}
R_3 &\ge \frac{1}{2}\log\frac{\sigma_{Z|Y}^2}{D_{Z_1}}, \\
R_4 + R_5 &\ge \frac{1}{2}\log\frac{\sigma_{Z|Y}^2}{D_{Z_2}}.
\end{align*}

From~\cite{Steinberg--Merhav}, the following rates are achievable
\begin{align*}
R_3 &= I(U_1, U_2;Z|Y), \\
R_4 &= I(U_2;Z|Y), \\
R_5 &= I(U_3, U_1;Z|Y, U_2).
\end{align*}
First, assume $R_3 \le \frac{1}{2}\log\frac{\sigma_{Z|Y}^2}{D_{Z_2}}$. Choose $ D_{Z_2}\le D' \le D'' \le D_{Z_1}$. We choose the auxiliary random variables and reconstruction functions as follows.
\begin{align*}
U_3 &= Z + W_3 \mbox{ where } W_3 \sim N(0, Q(D_{Z_2})),\\
U_1 &= U_3 + W_1 \mbox{ where } W_1 \sim N(0, Q(D')-Q(D_{Z_2})),\\
U_2 &= U_1 + W_2 \mbox{ where } W_2 \sim N(0, Q(D'')-Q(D')),\\
\Zh_1 &= \E(Z|U_1, Y), \\
\Zh_2 &= \E(Z|U_3, Y).
\end{align*}
From this choice of auxiliary random variables, it is easy to verify the following{\allowdisplaybreaks
\begin{align*}
R_3 & = I(U_1, U_2;Z|Y) \\
& = I(U_1;Z|Y) \\
&= \frac{1}{2}\log\frac{\sigma_{Z|Y}^2}{D'}, \\
R_4 &= I(U_2;Z|Y) \\
& = \frac{1}{2}\log\frac{\sigma_{Z|Y}^2}{D''}, \\
R_4 + R_5 & = I(U_2;Z|Y) + I(U_3, U_1;Z|Y, U_2) \\
& = I(U_3,U_1, U_2;Z|Y) \\
& = I(U_3;Z|Y) \\
&= \frac{1}{2}\log\frac{\sigma_{Z|Y}^2}{D_{Z_2}}, \\
\E(Z - \Zh_1)^2 &= D' \le D_{Z_1}, \\
\E(Z - \Zh_2)^2 &= D_{Z_2}.
\end{align*} }

Next, consider $R_3 > \frac{1}{2}\log\frac{\sigma_{Z|Y}^2}{D_{Z_2}}$ and $R_4 > \frac{1}{2}\log\frac{\sigma_{Z|Y}^2}{D_{Z_2}}$. Then, it is easy to see from our achievability scheme that we can obtain $R_4' < R_4$, $R_3'< R_3$ and $R_5 = 0$ by setting $D' = D'' = D_{Z_2}$. Finally, consider the case where $R_3 > \frac{1}{2}\log\frac{\sigma_{Z|Y}^2}{D_{Z_2}}$ and $R_4 \le \frac{1}{2}\log\frac{\sigma_{Z|Y}^2}{D_{Z_2}}$. Then, we observe from our achievability scheme that we can achieve $R_3' = \frac{1}{2}\log\frac{\sigma_{Z|Y}^2}{D_{Z_2}}< R_3$ for any $R_4$ and $R_5$ satisfying the inequalities by setting $D' = D_{Z_2}$.

{\em Case 3: $D_{Z_1} > D_{Z_2}$, $R_3 < R_4$}

In this case, the active inequalities are
\begin{align*}
R_3 &\ge \frac{1}{2}\log\frac{\sigma_{Z|Y}^2}{D_{Z_1}}, \\
R_3 + R_5 &\ge \frac{1}{2}\log\frac{\sigma_{Z|Y}^2}{D_{Z_2}}.
\end{align*}
We first consider the case where $R_3 \le \frac{1}{2}\log\frac{\sigma_{Z|Y}^2}{D_{Z_2}}$. We exhibit a scheme for which $R_4' = R_3 \, (< R_4)$ and still satisfies the constraints. This procedure is done by letting $U_2$ in case 2 to be equal to $U_1$. For $D_{Z_2} \le D' \le D_{Z_1}$, define the auxiliary random variables and reconstruction functions as follows. {\allowdisplaybreaks
\begin{align*}
U_3 &= Z + W_3 \mbox{ where } W_3 \sim N(0, Q(D_{Z_2})),\\
U_1 &= U_3 + W_1 \mbox{ where } W_1 \sim N(0, Q(D')-Q(D_{Z_2})),\\
\Zh_1 &= \E(Z|U_1, Y), \\
\Zh_2 &= \E(Z|U_3, Y).
\end{align*} }
Then, we have the following.{\allowdisplaybreaks
\begin{align*}
R_3 & = I(U_1;Z|Y) \\
&= \frac{1}{2}\log\frac{\sigma_{Z|Y}^2}{D'}, \\
R_4' &= I(U_1;Z|Y) \\
& = \frac{1}{2}\log\frac{\sigma_{Z|Y}^2}{D'}, \\
R_3 + R_5 & = I(U_1;Z|Y) + I(U_3;Z|Y, U_1) \\
& = I(U_3,U_1;Z|Y) \\
& = I(U_3;Z|Y) \\
&= \frac{1}{2}\log\frac{\sigma_{Z|Y}^2}{D_{Z_2}}, \\
\E(Z - \Zh_1)^2 &= D' \le D_{Z_1}, \\
\E(Z - \Zh_2)^2 &= D_{Z_2}.
\end{align*} }

Finally, we note that in the case where $R_3 > \frac{1}{2}\log\frac{\sigma_{Z|Y}^2}{D_{Z_2}}$, we can always achieve $R_3' = \frac{1}{2}\log\frac{\sigma_{Z|Y}^2}{D_{Z_2}}$, $R_4' = \frac{1}{2}\log\frac{\sigma_{Z|Y}^2}{D_{Z_2}}$ and $R_5' = 0$ by letting $D' = D_{Z_2}$. 
\end{proof}
{\em Remark 1: } The Two-way Cascade source coding setup given in section~\ref{sect:2} can be obtained as a special case by setting $R_3 = R_4 = 0$ and $D_{Z_1} \to \infty$.

{\em Remark 2: } The rate distortion region is the same regardless of whether Node 2 sends first, or Node 0 sends first. This observation follows from (i) our result in Corollary~\ref{coro1} where we showed that the rate distortion region for the Cascade setup is equivalent to the setup where all nodes have the degraded side information $Z$; and (ii) our proof above where we showed that the backward rates are the same as in the case where the side information $(X,Y)$ is available at all nodes.

{\em Remark 3: } For arbitrary sources and distortions, the problem is open in general. Even in the Gaussian case, the problem is open without the Markov Chain $X-Y-Z$. One may also consider the setting where there is a triangular source coding setup in the forward path from Node 0 to Node 2. This setting is still open, since the trade off in sending from Node 0 to Node 2 and then to Node 1 versus sending directly to Node 1 from Node 0 is not clear. 
\section{Triangular Source Coding with a helper} \label{sect:5}

We present an extension to our Triangular source coding setup by also allowing the side information $Y$ to be observed at the second node through a rate limited link (or helper). The setup is shown in Figure~\ref{fig3}. As the formal definitions are natural extensions of those given in section~\ref{sect:2}, we will omit them here.

\begin{figure} [!h] 
\begin{center}
\psfrag{x}{$X$}
\psfrag{y}{$Y$}
\psfrag{z}{$Z$}
\psfrag{r1}{$R_1$}
\psfrag{r2}{$R_2$}
\psfrag{r3}{$R_3$}
\psfrag{rh}{$R_h$}
\psfrag{x1}{$\Xh_1$}
\psfrag{x2}{$\Xh_2$}
\psfrag{n0}[l]{Node 0}
\psfrag{n1}[l]{Node 1}
\psfrag{n2}[l]{Node 2}
\psfrag{h}[c]{Helper}
\scalebox{0.75}{\includegraphics{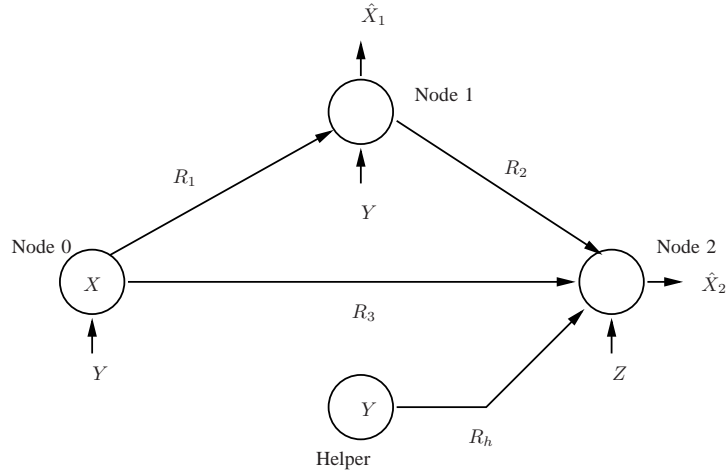}} 
\end{center}
\caption{Triangular Source Coding with a helper} \label{fig3}
\end{figure} 

\begin{theorem}
The rate distortion region for Triangular source coding with a helper is given by the set of rate tuples
\begin{align*}
R_1 &\ge I(X;\Xh_1, U_1|Y, U_h), \\
R_2 &\ge I(U_1;X,Y|Z, U_h), \\
R_3 &\ge I(X,Y;U_2|U_1,U_h, Z), \\
R_h &\ge I(U_h;Y|Z).
\end{align*}
for some $p(x,y,z,u_1,u_2,u_h, \xh_1) = p(x)p(y|x)p(z|y)p(u_h|y)p(u|x,y, u_h)p(\xh_1|x,y,u_1, u_h)p(u_2|x,y,u_1, u_h)$ and function $g_2 : \Uc_1 \times \Uc_2 \times \Uc_h\times \Zc\rightarrow \Xh_2$ such that
\begin{align*}
\E d_j(X_j, \Xh_j) \le D_j, \quad \mbox{j=1,2}.
\end{align*}
\end{theorem}

We give a proof of the converse. As the achievability techniques used form a straightforward extension of the techniques described in Appendix~\ref{appen:1}, we give only a sketch of achievability.

\begin{proof}[Converse]
Given a $(n, 2^{nR_1}, 2^{nR_2}, 2^{nR_3}, 2^{nR_h}, D_1, D_2)$ code, define $U_{hi} = (Y^{i-1}, Z^{i-1}, Z_{i+1}^n,M_h)$, $U_{1i} = (X^{i-1}, M_2)$ and $U_{2i} = (U_{hi}, U_{1i}, M_3)$. Observe that we have the required Markov conditions $(X_i,Z_i) - Y_i - U_{hi}$ and $Z_i - (X_i, Y_i, U_{hi}) - (U_{1i}, U_{2i})$.
For the helper condition, we have 
\begin{align*}
nR_h &\ge I(M_h; Y^n|Z^n) \\
& = \sum_{i=1}^n H(Y_i|Z_i) - H(Y_i|Y^{i-1}, M_h, Z^n) \\
& = \sum_{i=1}^n I(U_{hi};Y_i|Z_i).
\end{align*}
For the other rates, we have {\allowdisplaybreaks
\begin{align*}
nR_1 &\ge H(M_1) \\
&\ge H(M_1|Y^n, Z^n) \\
& = H(M_1,M_2|Y^n, Z^n) = I(X^n; M_1,M_2|Y^n, Z^n) \\
& =\sum_{i=1}^n I(X_i; M_1, M_2|X^{i-1}, Y^n, Z^n) \\
& =\sum_{i=1}^n H(X_i|X^{i-1}, Y^n, Z^n) - H(X_i, Y_i|X^{i-1}, Y^n, Z^n, M_1, M_2) \\
& =\sum_{i=1}^n H(X_i|Y_i, Z_i) - H(X_i, Y_i|X^{i-1}, Y^n, Z^n, M_1, M_2) \\
& \stackrel{(a)}{=}\sum_{i=1}^n H(X_i|Y_i) - H(X_i, Y_i|X^{i-1}, Y^n, \Xh_{1i}, Z^n, M_1, M_2, M_h) \\
& \ge \sum_{i=1}^n H(X_i|Y_i, U_{hi}) - H(X_i|\Xh_{1i}, Y_i, U_{1i}, U_{hi}) \\
& = \sum_{i=1}^n I(X_i; \Xh_{1i}, U_{1i}|Y_i, U_{hi}).
\end{align*} }
$(a)$ follows from the Markov chain condition. Next,{\allowdisplaybreaks
\begin{align*}
nR_2 &\ge H(M_2|M_h) \\
&\ge H(M_2|Z^n, M_h) \\
& = I(X^n, Y^n; M_2|Z^n, M_h) \\
& = \sum_{i=1}^n I(X_i, Y_i; M_2|Z^n, X^{i-1}, Y^{i-1}, M_h) \\
& = \sum_{i=1}^n H(X_i, Y_i|Z^n, X^{i-1}, Y^{i-1},M_h) - H(X_i, Y_i|Z^n, X^{i-1}, Y^{i-1}, M_2,M_h)\\
& = \sum_{i=1}^n H(X_i, Y_i|Z_i,U_{hi}) - H(X_i, Y_i|Z_i, U_{1i},U_{hi})\\
& = \sum_{i=1}^n I(X_i, Y_i; U_{1i}|Z_i,U_{hi}).
\end{align*}}
Next,
\begin{align*}
nR_3 &\ge H(M_3) \\
&\ge H(M_3|M_2,M_h, Z^n)\\
& = I(X^n, Y^n; M_3|M_2,M_h Z^n) \\
& = \sum_{i=1}^n H(X_i, Y_i|M_2, M_h, Z^n, X^{i-1}, Y^{i-1}) - H(X_i, Y_i|M_2, M_3, M_h, Z^n, X^{i-1}, Y^{i-1}) \\
& = \sum_{i=1}^n I(X_i, Y_i; U_{2i}|U_{1i}, U_{hi}, Z_i).
\end{align*}
Finaly, it remains to show that the joint probability distribution induced by our choice of auxiliary random variables, $p(x)p(y|x)p(z|y)p(u_h|y)p(u|x,y, u_h)p(\xh_1, u_2|x,y,u_1, u_h)$, can be decomposed into the required form. This step follows closely the similar step in the proof of Theorem~\ref{thm2}, which we therefore omit.

\noindent {\em Sketch of Achievability}

The achievability follows that of Triangular source coding, with an additional step of generating a lossy description of $Y^n$. The codebook generation consists of the following steps
\begin{itemize}
\item Generate $2^{n(I(Y;U_h)+\e)}$ $U_h^n$ sequences according to $\prod_{i=1}^n p(u_{hi})$. Partition the set of $U_h^n$ sequences into $2^{n(I(U_h;Y|Z)+2\e)}$ bins, $\Bc_h(m_h)$, $m_h \in [1: 2^{n(I(U_h;Y|Z)+2\e)}]$. 

\item Generate $2^{n(I(X,Y, U_h;U_1)+\e)}$ $U_1^n$ sequences according to $\prod_{i=1}^n p(u_{1i})$. Partition the set of $U_1^n$ sequences into $2^{n(I(U_1;X|Y, U_h)+2\e)}$ bins, $\Bc_1(m_{10})$. Separately and independently, partition the set of $U^n$ sequences into $2^{n(I(U_1;X,Y|Z, U_h)+2\e)}$ bins, $\Bc_2(m_2)$, $m_2 \in [1:2^{n(I(U;X,Y|Z)+2\e)}]$. 

\item For each $(u_1^n, u_h^n, y^n)$ sequence, generate $2^{n(I(\Xh^n_1; X|U_1,Y, U_h)+\e)}$ $\Xh_1^n$ sequences according to $\prod_{i=1}^n p(\xh_i|u_{1i},u_{hi}, y_i)$.

\item Generate $2^{n(I(U_2;X,Y|U_h, U_1)+ \e)}$ $U_2^n$ sequences according to $\prod_{i=1}^n p(u_{2i}|u_{1i}, u_{hi})$ for each $(u_1^n, u_h^n)$ sequence, and partition these sequences to $2^{n(I(U_2;X,Y|U_1, U_h, Z)+ 2\e)}$ bins, $\Bc_3(m_3)$. 
\end{itemize}

Encoding consists of the following steps
\begin{itemize}
\item  Helper node: The helper node (and Nodes 0 and 1) looks for a $u_h^n$ sequence such that $(u_h^n, y^n) \in \aep$. This step succeeds with high probability since there are $2^{n(I(Y;U_h)+\e)}$ $U_h^n$ sequences. The helper then sends out the bin index $m_h$ such that $u_h^n \in \Bc(m_h)$. The sequences $(u_h^n, x^n,y^n,z^n)$ are jointly typical with high probability due to the Markov Chain $(X,Z) - Y-U_h$. 

\item Node 0: Given $(x^n,y^n, u^n_h) \in \aep$, Node 0 looks for a jointly typical codeword $u_1^n$. This operation succeeds with high probability since there are $2^{n(I(X,Y, U_h;U_1)+\e)}$ $U_1^n$ sequences. Node 0 then looks for a $\xh_1^n$ that is jointly typical with $(u^n_1, x^n, y^n, u_h^n)$. This operation succeeds with high probability since there are $2^{n(I(\Xh_1; X|U_1,U_h,Y)+\e)}$ $\xh_1^n$ sequences. 

\item Node 0 also finds a $u_2^n$ sequence that is jointly typical with $(u_1^n, u_h^n,x^n, y^n)$. This operation succeeds with high probability since we have $2^{n(I(U_2;X,Y|U_1, U_h)+ \e)}$ $v^n$ sequences. 

\item Node 0 then sends out the bin index $m_{10}$ such that $u_1^n \in \Bc_1(m_{10})$ and the index corresponding to $\xh_1^n$ to Node 1.  This requires a total rate of $R_1 = I(U;X|Y)+ I(\Xh^n_1; X|U,Y)+3\e$ to Node 1. Node 0 also sends out the bin index $m_3$ such that $u_2^n \in \Bc(m_3)$ to Node 2. This requires a rate of $I(U_2;X,Y|U_1, U_h, Z)+ 2\e$.

\item Node 1 decodes the codeword $u_1^n$ and forwards the index $m_2$ such that $u_1^n \in \Bc(m_2)$ to Node 2. This requires a rate of $I(U_1;X,Y|Z, U_h)+2\e$.
\end{itemize}

Decoding consists of the following steps

\begin{itemize}
\item Node 1: Node 1 reconstructs $u_1^n$ by looking for the unique $U_1^n$ sequence in $\Bc_1(m_{10})$ such that $(U_1^n, U_h^n, Y^n) \in \aep$. Since there are only $2^{n(I(X,Y, U_h;U_1) - I(U_1;X|Y, U_h)-\e)} = 2^{n(I(U_1;U_h,Y)-\e)}$ sequences in the bin, this operation succeeds with high probability. Node 1 reconstructs $X^n$ as $\Xh_1^n(m_{10}, m_{11})$. Since the sequence $(\Xh_1^n, X^n)$ are jointly typical with high probability, the expected distortion constraint is satisfied.

\item Node 2: We note that since $(U_1, U_2, U_h, X)-Y-Z$, the sequences $(U_h^n, U_1^n, U_2^n, X^n, Y^n, Z^n)$ are jointly typical with high probability. Decoding at node 2 consists of the following steps
\begin{enumerate}
\item Node 2 first looks for $u_h^n$ in $\Bc_h(m_h)$ such that $(u_h^n,z^n) \in \aep$. This operation succeeds with high probability since there are only $2^{n(I(U_h;Z)-\e}$ $u_h^n$ sequences in the bin. 
\item It then looks for $u_1^n$ in $\Bc_2(m_2)$ such that $(u_h^n,u_1^n, z^n) \in \aep$. Since $I(U_1; X,Y, U_h) - I(U_1;X,Y|Z,U_h) = I(U_1; Z, U_h)$ by the Markov Chain $Z - (X,Y,U_h) - U_1$, this operation succeeds with high probability as there are only $2^{n(I(U_1;Z, U_h)-\e}$ $u_1^n$ sequences in the bin. 
\item Finally, it looks for $u_2^n$ in $\Bc_3(m_3)$ such that $(u_h^n,u_1^n, u_2^n, z^n) \in \aep$. Since $I(U_2; X,Y|U_h, U_1) - I(U_2;X,Y|Z,U_h, U_1) = I(U_2; Z |U_1, U_h)$ by the Markov Chain $Z - (X,Y,U_h, U_1) - U_2$, this operation succeeds with high probability as there are only $2^{n(I(U_2;Z| U_1,U_h)-\e}$ $u_2^n$ sequences in the bin. 
\item Node 2 then reconstructs using the function $\xh_{2i} = g_2(u_{1i}, u_{2i}, u_{hi}, z_i)$ for $i \in [1:n]$. Since the sequences $(X^n, Z^n, U_1^n, U_2^n, U_h^n)$ are jointly typical with high probability, the expected distortion constraint is satisfied.
\end{enumerate}
\end{itemize}
\end{proof}

\section{Conclusion} \label{sect:6}
Rate distortion regions for the cascade, triangular, two-way cascade and two-way triangular source coding settings were established. Decoding part of the description intended for Node 2 and then re-binning it was shown to be optimum for our Cascade and Triangular settings. We also extended our Triangular setting to the case where there is an additional rate constrained helper, which observes $Y$, for Node 2. In the Quadratic Gaussian case, we showed that the auxiliary random variables can be taken to be jointly Gaussian and that the rate-distortion regions obtained for the Cascade and Triangular setup were equivalent to the setting where the degraded side information is available at all nodes. This observation allows us to transform our Cascade and Triangular settings into equivalent settings for which explicit characterizations are known. Characterizations of the rate distortion regions for the Quadratic Gaussian cases were also established in the form of tractable low dimensional optimization programs. Our Two Way Cascade Quadratic Gaussian setting was extended to solve a more general two way cascade scenario. The case of generally distributed $X,Y,Z$, without the degradedness assumption, remains open.

\bibliographystyle{IEEEtran}
\bibliography{Cascade}
\appendices
\section{Achievability proofs} \label{appen:1}
\noindent \textit{Achievability proof of Theorem~\ref{thm1}}
\subsection{Codebook Generation}
\begin{itemize}
\item Fix the joint distribution $p(x,y,z,u,\xh_1) = p(x)p(y|x)p(z|y)p(u|x,y)p(\xh_1|x,y,u)$. Let $R = R_{10} + R_{11}$, $R_{l} \ge R_{10}$ and $R_{2} \ge R_{10}$.
\item Generate $2^{nR_{10}}$ $U^n(l)$ sequences, $l \in [1:2^{nR_{1}}]$, each according to $\prod_{i=1}^n p(u_i)$.
\item Partition the set of $U^n$ sequences into $2^{nR_{10}}$ bins, $\Bc_1(m_{10})$, $m_{10} \in [1:2^{nR_{10}}]$. Separately and independently, partition the set of $U^n$ sequences into $2^{nR_2}$ bins, $\Bc_2(m_2)$, $m_2 \in [1:2^{nR_2}]$.
\item For each $u^n(l)$ and $y^n$ sequences, generate $2^{nR_{11}}$ $\Xh^n_1(l, m_{11})$ sequences according to $\prod_{i=1}^n p(\xh_{1i}|u_i, y_i)$. 
\end{itemize}
\subsection{Encoding at the encoder}
Given a $(x^n, y^n)$ pair, the encoder first looks for an index $l \in [1:2^{nR_{l}}]$ such that $(u^n(l), x^n, y^n) \in \aep$, where $\aep$ stands for the set of jointly typical sequences. If there are more than one such $l$, it selects one uniformly at random from the set of admissible indices. If there is none, it sends an index uniformly at random from $[1:2^{nR_{l}}]$\footnote{For simplicity, we assume randomized encoding, but it is easy to see that the randomized encoding employed our proofs can be incorporated as part of the (random) codebook generation stage.}. Next, it finds the index $m_{11}$ such that $(\xh_1(l, m_{11}), u^n(m_{10}), x^n, y^n) \in \aep$. As before, if there is more than one, it selects one uniformly at random from the set of admissible indices. If there is none, it sends an index uniformly at random from $[1:2^{nR_{11}}]$. Finally, it sends out $(m_{10},m_{11})$, where $m_{10}$ is the bin index such that $u^n(l) \in \Bc_1(m_{10})$. The total rate required is $R$.
\subsection{Decoding and reconstruction at Node 1}
Given $(m_{10},m_{11})$, Node 1 looks for the unique $\lh$ such that $(u^n(\lh), y^n) \in \aep$ and $u^n(\lh) \in \Bc_1(l)$. It reconstructs $x^n$ as $\xh^n(\lh, m_{11})$. If it failed to find a unique one, or if there is more than one, it outputs $\lh = 1$ and performs the reconstruction as before.
\subsection{Encoding at Node 1}
Node 1 sends an index $\mh_2$ such that $u^n(\lh) \in \Bc_2(\mh_2)$. This requires a rate of $R_2$.
\subsection{Decoding and reconstruction at Node 2}
Node 2 looks for the index $\lt$ such that $(u^n(\lt), y^n) \in \aep$ and $\lt \in \Bc_2(\mh_2)$. It then reconstructs $x^n$ according to $\xh_{2i} = g_2(u^n(\lt)_i, z_i)$ for $i\in [1:n]$. If there is no such index, it reconstructs using $\lt = 1$.
\subsection{Analysis of expected distortion}
Using the typical average lemma in~\cite[Lecture 2]{El-Gamal--Kim2010} and following the analysis in~\cite[Lecture 3]{El-Gamal--Kim2010}, it suffices to analyze the probability of ``error''; i.e. the probability that the chosen sequences will not be jointly typical with the source sequences. Let $L$ and $M_{11}$ be the chosen indices at the encoder. Note that these define the bin indices $M_{10}$ and $M_2$. Let $\Mh_2$ be the chosen index at Node 1. Define the following error events:
\begin{enumerate}
\item $\Ec_0 := \{(X^n, Y^n) \notin \aep\}$
\item $\Ec_1 := \{(U^n(l), X^n, Y^n) \notin \aep\}$ for all $l\in [1:2^{nR_l}]$
\item $\Ec_2 := \{(U^n(l), X^n, Y^n, Z^n) \notin \aep\}$ for all $l\in [1:2^{nR_l}]$
\item $\Ec_3 := \{(U^n(L), \Xh^n(L, m_{11}), X^n, Y^n) \notin \aep\}$ for all $m_{11}\in [1:2^{nR_{11}}]$
\item $\Ec_4 := \{(U^n(\lh), Y^n) \in \aep\}$ for some $\lh \neq L$ and $U^n(\lh) \in \Bc_1(M_{10})$
\item $\Ec_5(\Mh_2) := \{(U^n(\lt), Z^n) \in \aep\}$ for some $\lt \neq L$ and $U^n(\lt) \in \Bc_2(\Mh_2)$
\end{enumerate}
We can then bound the probability of error as
\begin{align*}
\P_e &\le \P\{\bigcup_{i=0}^5 \Ec_i\} = \sum \P\{\Ec_i \cap (\bigcap_{j=0}^{i-1}\Ec_j^c)\}.
\end{align*}
\begin{itemize}
\item $\P\{\Ec_0\}\to 0$ as $n \to \infty$ by Law of Large Numbers (LLN).
\item By the covering lemma in~\cite[Lecture 3]{El-Gamal--Kim2010}, $\P\{\Ec_1\cap \Ec_0^c\}\to 0$ as $n \to \infty$ if
\begin{align*}
R_{l} > I(U;X,Y) + \d(\e).
\end{align*}
\item $\P\{\Ec_2 \cap \Ec_1^c\cap \Ec_0^c\}\to 0$ as $n \to \infty$ by the Markov relation $U-(X,Y) - Z$ and the conditional joint typicality lemma~\cite[Lecture 2]{El-Gamal--Kim2010}.
\item By the covering lemma in~\cite[Lecture 3]{El-Gamal--Kim2010}, $\P\{\Ec_3\cap (\bigcap_{j=0}^{2}\Ec_j^c\}\to 0$ as $n \to \infty$ if
\begin{align*}
R_{11} > I(\Xh_1;X|U,Y) + \d(\e).
\end{align*}
\item From the analysis of the Wyner-Ziv Coding scheme (see~\cite{Wyner} or~\cite[Lecture 12]{El-Gamal--Kim2010}), $\P\{\Ec_4\cap (\bigcap_{j=0}^{3}\Ec_j^c\}\to 0$ as $n \to \infty$ if
\begin{align*}
R_l - R_{10} < I(U;Y) - \d(\e).
\end{align*}
\item For the last term, we have
\begin{align*}
\P\{\Ec_5(\Mh_2)\cap (\bigcap_{j=0}^{4}\Ec_j^c)\} & {=} \P\{\Ec_5(\Mh_2)\cap (\bigcap_{j=0}^{4}\Ec_j^c) \cap \{\Mh_2 \neq M_2\}\} \\
&\qquad + \P\{\Ec_5(\Mh_2)\cap (\bigcap_{j=0}^{4}\Ec_j^c) \cap \{\Mh_2 = M_2\}\} \\
& \stackrel{(a)}{=} \P\{\Ec_5(\Mh_2)\cap (\bigcap_{j=0}^{4}\Ec_j^c) \cap \{\Mh_2 = M_2\}\} \\
& = \P\{\Ec_5(M_2)\cap (\bigcap_{j=0}^{4}\Ec_j^c) \cap \{\Mh_2 = M_2\}\} \\
& \le  \P\{\Ec_5(M_2)\cap \Ec_{2}^c\}.
\end{align*}
Step $(a)$ follows from the observation that $(\bigcap_{j=0}^{4}\Ec_j^c) \cap \{\Mh_2 \neq M_2\} = \emptyset$. The analysis of the probability of error therefore reduces to the analysis for the equivalent Wyner-Ziv setup with $Z$ as the side information at Node 2. Hence, $\P\{\Ec_5(\Mh_2)\cap (\bigcap_{j=0}^{4}\Ec_j^c)\}\to 0$ as $n\to \infty$ if
\begin{align*}
R_{l} - R_{2} < I(U;Z) - \d(\e).
\end{align*}
\end{itemize}
Eliminating $R_{l}$ in the aforementioned inequalities then gives us the required rate region. 

\noindent \textit{Achievability proof of Theorem~\ref{thm2}}

As the achievability proof for the Triangular Source Coding Case follows that of the Cascade Source Coding Case closely, we will only include the additional steps required for generating $R_3$ and analysis of probability of error at Node 2. The steps for generating $R_1$ and $R_2$, and for reconstruction at Node 1 are the same as the Cascade setup.

\subsection{Codebook Generation}
\begin{itemize}
\item Fix $p(x,y,z,u,v, \xh_1) = p(x)p(y|x)p(z|y)p(u|x,y)p(\xh_1|x,y,u)p(v|x,y,u)$. 
\item For each $u^n(l)$, generate $V^n(l_3)$, $l_3 \in [1:2^{n\Rt_3}]$, according to $\prod_{i=1}^n p(v_i|u_i)$. Partition the set of $v^n$ sequences into $2^{nR_3}$ bins, $\Bc_3(m_3)$. 
\end{itemize}
\subsection{Encoding}
\begin{itemize}
\item Given a sequence $(x^n, y^n)$ and $u^n(l)$ found through the steps in the Cascade Source Coding setup, the encoder looks for an index $l_3$ such that $(u^n, v^n(l, l_3), x^n, y^n) \in \aep$. If it finds more than one, it selects one uniformly at random from the set of admissible indices. If it finds none, it outputs an index uniformly at random from $[1:2^{n\Rt_3}]$. The encoder then sends out $m_3$ such that $L_3 \in \Bc_3(m_3)$.
\end{itemize}
\subsection{Decoding}
The additional decoding step is in decoding $L_3$. Node 2 looks for the unique $\lh_3$ such that $(u^n(\lt), v^n(\lt, \lh_3),z^n) \in \aep$ and $v^n(\lh_3) \in \Bc_3(M_3)$. If there is none or more than one, it outputs $\mh_3 = 1$.
\subsection{Analysis of Distortion}
Let $L$, $M_{11}$ and $M_3$ be the indices chosen by the encoder. Note that these fix the indices $M_{10}$ and $M_2$. We follow similar analysis as in the Cascade case, with the same definitions for error events $\Ec_0$ to $\Ec_5$. We also require the following additional error events:
\begin{enumerate}
\item[7)] $\Ec_{6}:=\{(U^n(L), V^n(L, L_{3}), X^n, Y^n)\notin \aep\}$.
\item[8)] $\Ec_{7}:=\{(U^n(L), V^n(L, L_{3}), X^n, Y^n, Z^n)\notin \aep\}$.
\item[9)] $\Ec_8(\Lt):= \{(U^n(\Lt), V^n(\Lt, \lh_{3}), Z^n) \in \aep\}$ for some $\lh_3 \neq L_3$ and $\lh_3 \in \Bc_3(M_3)$.
\end{enumerate}
To bound the probability of error, we have the following additional terms
\begin{itemize}
\item By the covering lemma, $\P(\Ec_6\cap \Ec_2^c)\to 0$ as $n \to \infty$ if
\begin{align*}
\Rt_3 > I(V;X,Y|U) + \d(\e).
\end{align*}
\item $\P(\Ec_7\cap \Ec_6^c)\to 0$ as $n \infty$ from the Markov condition $(V,U) - (X,Y) - Z$ and the conditional joint typicality lemma.
\item $\P\{\Ec_8(\Lt) \cap\Ec^c_5(\Mh_2)\cap \Ec^c_7 \cap (\bigcap_{j=0}^{4}\Ec_j^c))\}$. We have{\allowdisplaybreaks
\begin{align*}
& \P\{\Ec_8(\Lt) \cap\Ec^c_5(\Mh_2)\cap \Ec^c_7 \cap (\bigcap_{j=0}^{4}\Ec_j^c)\} \\
& =\P\{\Ec_8(\Lt) \cap\Ec^c_5(\Mh_2)\cap \Ec^c_7 \cap (\bigcap_{j=0}^{4}\Ec_j^c) \cap \{\Mh_2 = M_2\}\} + \P\{\Ec_8(\Lt) \cap\Ec^c_5(\Mh_2)\cap \Ec^c_7 \cap (\bigcap_{j=0}^{4}\Ec_j^c)\cap\{\Mh_2 \neq M_2\}\} \\
& = \P\{\Ec_8(\Lt) \cap\Ec^c_5(\Mh_2)\cap \Ec^c_7 \cap (\bigcap_{j=0}^{4}\Ec_j^c) \cap \{\Mh_2 = M_2\}\} \\
& = \P\{\Ec_8(\Lt) \cap\Ec^c_5(M_2)\cap \Ec^c_7 \cap (\bigcap_{j=0}^{4}\Ec_j^c) \cap \{\Mh_2 = M_2\}\} \\
& \le \P\{\Ec_8(\Lt) \cap\Ec^c_5(M_2)\cap \Ec^c_7\} \\
& \stackrel{(a)}{=} \P\{\Ec_8(\Lt) \cap\Ec^c_5(M_2)\cap \Ec^c_7 \cap \{\Lt = L\}\} \\
& = \P\{\Ec_8(L) \cap\Ec^c_5(M_2)\cap \Ec^c_7 \cap \{\Lt = L\}\} \\
& \le \P\{\Ec_8(L)\cap \Ec_7^c\}.
\end{align*}}
$(a)$ follows from the observation that $\Ec^c_5(M_2)\cap \Ec^c_7 \cap \{\Lt \neq L\} = \emptyset$. It remains to bound $\P\{\Ec_8(L)\cap \Ec_7^c\}$. Note that the analysis of this term is equivalent to analyzing the setup where $U^n$ is the side information at Node 0 and $(U^n, Z^n)$ is the side information at Node 2. Hence, $\P\{\Ec_8(L)\cap \Ec_7^c\} \to 0$ as $n \to \infty$ if
\begin{align*}
\Rt_3 - R_3 < I(V;Z|U) - \d(\e).
\end{align*}
\end{itemize}
We then obtain the rate region by eliminating $\Rt_3$ and $R_{l}$.

\noindent \textit{Achievability proof of Theorem~\ref{thm3}}

As with the case for the Triangular setting, the proof for this case follows the Cascade setting closely. We will therefore include only the additional steps. We have a change of notation from the Cascade setting. We will use $U_1$ instead of $U$
\subsection{Codebook Generation}
\begin{itemize}
\item Fix $p(x,y,z, u_1, u_2, \xh_1)= p(x,y,z)p(u_1|x,y)p(\xh_1|u_1,x,y)p(u_2|z,u_1)$. 
\item For each $u^n_1(l)$, generate $2^{nR_3}$ $U_2^n(l_3)$ sequences, $l \in [1:2^{n\Rt_3}]$, each according to $\prod_{i=1}^np(u_{2i}|u_{1i})$. Partition the set of $U^n_2$ into $2^{nR_3}$ bins, $\Bc_3(m_3)$.
\end{itemize}
\subsection{Encoding}
The additional encoding step is at Node 2. Node 2 looks for an index $L_3$ such that $(u^n_1(L), u^n_2(L, L_3), Z^n) \in \aep$. As before, if it finds more than one, it selects an index uniformly at random from the set of admissible indices. If it finds none, it outputs an index uniformly at random from $[1:2^{n\Rt_3}]$. It then outputs the bin index $m_3$ such that $L_3 \in \Bc_3(m_3)$. 
\subsection{Decoding}
Additional decoding is required at Node 0. Node 0 looks the index $\lh_3$ such that $(u_1^n(l), u_2^n(l, \lh_3), x^n, y^n) \in \aep$ and $\lh_3 \in \Bc_3(m_3)$. 
\subsection{Analysis of distortion}
Let $\Ec_{Cascade}$ denote the event that an error occurs in the forward Cascade path. In addition, we define the following error events.
\begin{itemize}
\item $\Ec_{TW-1}(\Lh):= \{(U^n_1(\Lh), U^n_2(\Lh, l_3), Z^n) \notin \aep$ for all $l_3 \in [1:2^{n\Rt_3}]\}$.
\item $\Ec_{TW-2}(\Lh):= \{(U^n_1(\Lh), U^n_2(\Lh, L_3), Z^n, X^n, Y^n) \notin \aep\}$.
\item $\Ec_{TW-3}(\Lh):= \{(U^n_1(\Lh), U^n_2(\Lh, \lh_3), X^n, Y^n) \in \aep$ for some $\lh_3 \in \Bc_3(M_3), \lh_3 \neq L_3\}$.
\end{itemize}
\begin{itemize}
\item $\P(\Ec_{TW-1}(\Lh) \cap \Ec_{Cascade}^c) = \P(\Ec_{TW-1}(L) \cap \Ec_{Cascade}^c)\to 0$ as $n \to \infty$ if
\begin{align*}
\Rt_3 > I(U_2;Z|U_1) + \d(\e).
\end{align*}
\item $\P(\Ec_{TW-2}(\Lh) \cap \Ec_{Cascade}^c) = \P(\Ec_{TW-2}(L) \cap \Ec_{Cascade}^c)\to 0$ as $n \to \infty$ by the strong Markov Lemma~\cite{Tung}.
\item $\P(\Ec_{TW-3}(\Lh) \cap \Ec_{Cascade}^c) = \P(\Ec_{TW-3}(L) \cap \Ec_{Cascade}^c)\to 0$ as $n \to \infty$ if
\begin{align*}
\Rt_3 - R_3 < I(U_2; X,Y|U_1) - \d(\e).
\end{align*}
\end{itemize}

Finally, eliminating $\Rt_3$ and $R_l$ gives us the required rate region.

\noindent \textit{Achievability proof of Theorem~\ref{thm4}}

The achievability proof for Two Way Triangular source coding combines the proofs of the Triangular source coding case and the Two-way cascade case. As it is largely similar to these proofs, we will not repeat it here. We will just mention that the codebook generation, encoding, decoding and analysis of distortion for the forward path from Node 0 to Node 2 follows that of the Triangular source coding case, while codebook generation, encoding, decoding and analysis of distortion for the reverse path from Node 2 to Node 0 follows that of the Two-way Cascade source coding case, with $(U_2, V)$ taking the role of $U_2$.

\section{Cardinality Bounds} \label{appen:2}
We provide cardinality bounds for Theorems~\ref{thm1}-\ref{thm4} stated in the paper. The main tool we will use is the Fenchel-Eggleston-Caratheodory Theorem~\cite{Eggleston}. 

\subsection{Proof of cardinality bound for Theorem~\ref{thm1}}

For each $x,y$, we have
\begin{align*}
f_{j}(p_{X,Y|U}(x,y|u)) = \sum_{u}p(u)p(x,y|u) = p(x,y). 
\end{align*}
We therefore have $|\Xc||\Yc| - 1$ continuous functions of $p(x,y|u)$. These set of equations preserves the distribution $p(x,y)$ and hence, by Markovity, $p(x,y,z)$. Next, observe that the following are similarly continuous functions of $p(x,y|u)$
\begin{align*}
I(U;X,Y|Z) &= H(X,Y|Z) - H(X,Y,Z|U) + H(Z|U), \\
I(X;\Xh_1, U|Y) &= H(X|Y) - H(X|U) + H(X, \Xh_1, Y|U), \\
\E d_1(X, \Xh_1) &= \sum_{x,\xh} p(x, \xh_1)d(x, \xh_1), \\
\E d_2(X, \Xh_2) &= \sum_{x,y,u} p(x, y,u)d(x, g_2(x,u)),
\end{align*}
These equations give us 4 additional continuous functions and hence, by Fenchel-Eggleston-Caratheodory Theorem, there exists a $U'$ with cardinality of $|\Xc||\Yc| + 3$ such that all the constraints are satisfied. Note that this construction does not preserve $p(\xh_1)$, but this does not change the rate-distortion region since the associated rate and distortion are preserved. 

\subsection{Proof of cardinality bound for Theorem~\ref{thm2}}
We will first give a bound for the cardinality of $U$. We look at the following continuous functions of $p(x,y|u)$.
\begin{align*}
f_{j}(p_{X,Y|U}(x,y|u)) &= \sum_{u}p(u)p(x,y|u) = p(x,y), \forall x,y \\ 
I(U;X,Y|Z) &= H(X,Y|Z) - H(X,Y,Z|U) + H(Z|U), \\
I(X;\Xh_1, U|Y) &= H(X|Y) - H(X|U) + H(X, \Xh_1, Y|U), \\
I(X,Y;V|U,Z) &= H(X,Y, Z|U) - H(Z|U) - H(X,Y,V,Z|U) + H(V,Z|U), \\ 
\E d_1(X, \Xh_1) &= \sum_{x,\xh} p(x, \xh_1)d(x, \xh_1), \\
\E d_2(X, \Xh_2) &= \sum_{x,y,u,v} p(x, y,u, v)d(x, g_2(x,u)).
\end{align*}
From these equations, there exists a $U'$ with $|\Uc'| \le |\Xc||\Yc| + 4$ such that the equations are satisfied. Note that the new $U'$ induces a new $V'$. For each $U' = u$, consider the following continuous functions of $p(x,y|u,v)$
\begin{align*}
p(x,y|u) &= \sum_{v}p(v|u)p(x,y|v,u), \\
I(X,Y;V|U=u, Z) &= H(X,Y|U=u,Z) - H(X,Y|V, U= u, Z), \\
\E(d_2(X, \Xh_2)|U=u) &= \sum_{x,y,v} p(x, y,v|u)d(x, g_2(x,u)).
\end{align*}
From this set of equations, we see that for each $U' = u$, it suffices to consider $V'$ such that $|\Vc'| \le |\Xc||\Yc| +1$. Hence, the overall cardinality bound on $V$ is $|\Vc|\le (|\Xc||\Yc| + 4)(|\Xc||\Yc| +1)$. The joint $p(x,y,z)$ is preserved due to the Markov Chain $(V,U)- (X,Y) - Z$.

\subsection{Proof of cardinality bound for Theorem~\ref{thm3}}
The cardinality bounds on $U_1$ follows similar analysis as in the Cascade source coding case. The proof is therefore omitted. For each $U_1 = u_1$, the following are continuous functions of $p(z|u_2, u_1)$,
\begin{align*}
p(z|u_1) &= \sum_{u_2}p(u_2|u_1)p(z|u_2, u_1), \\
I(U_2;Z|U_1=u_1, X,Y) &= H(Z|U_1=u_1, X,Y) - H(Z|U_1=u_1, U_2, X,Y), \\
\E(d_3(Z, \Zh)|U_1 = u_1) &= \sum_{x,y,z, u_2} p(x, y,z,u_2|u_1)d(z, g_3(x,y,u_1,u_2)). 
\end{align*}
From this set of equations, we see that for each $U_1 = u_1$, it suffices to consider $U_2'$ such that $|\Uc_2'| \le |\Zc|+1$. Hence, the overall cardinality bound on $U_2$ is $|\Uc_2|\le |\Uc_1|(|\Zc|+1)$. The joint $p(x,y,z)$ is preserved due to the Markov Chains $U_1 - (X,Y)- Z$ and $U_2- (Z,U_1) - (X,Y)$.

\subsection{Proof of cardinality bound for Theorem~\ref{thm4}}

The cardinality bounds follow similar steps to those for the first 3 theorems. For the cardinality bound for $|\Uc_2|$, we find a cardinality bound for each $U_1 = u_1$ and $V = v$. Details of the proof are omitted.

\section{Alternative characterizations of rate distortion regions in Corollaries~\ref{coro1} and~\ref{coro2}} \label{appen:3}
In this appendix, we show that the rate distortion regions in Corollaries~\ref{coro1} and~\ref{coro2} can alternatively be characterized by transforming them into equivalent problems found in~\cite{Permuter}, where explicit characterizations were given. We focus on the Cascade case (Corollary~\ref{coro1}), since the Triangular case follows by the same analysis.

Figure~\ref{fig6} shows the Cascade source coding setting which the optimization problem in Corollary~\ref{coro1} solves. 

\begin{figure} [!h] 
\begin{center}
\psfrag{x}[c]{$A+B$}
\psfrag{y}{$B$}
\psfrag{z}{}
\psfrag{r1}{$R_1$}
\psfrag{r2}{$R_2$}
\psfrag{x1}{$\Xh_1$}
\psfrag{x2}{$\Xh_2$}
\psfrag{n0}[l]{Node 0}
\psfrag{n1}[l]{Node 1}
\psfrag{n2}[l]{Node 2}
\scalebox{0.75}{\includegraphics{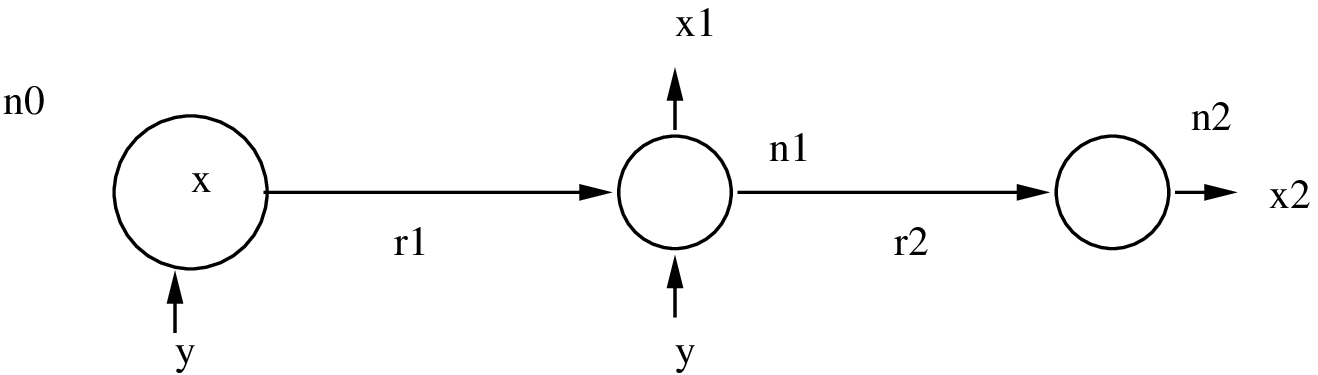}} 
\end{center}
\caption{Cascade source coding setting for the optimization problem in Corollary~\ref{coro1}. $\Xh_1$ and $\Xh_2$ are lossy reconstructions of $A+B$.} \label{fig6}
\end{figure} 

In~\cite{Permuter}, explicit characterization of the Cascade source coding setting in Figure~\ref{fig7} was given. 

\begin{figure} [!h] 
\begin{center}
\psfrag{x}[c]{$X$}
\psfrag{y}[c]{$Y = X+Z$}
\psfrag{z}{}
\psfrag{r1}{$R_1$}
\psfrag{r2}{$R_2$}
\psfrag{x1}{$\Xh_1$}
\psfrag{x2}{$\Xh_2$}
\psfrag{n0}[l]{Node 0}
\psfrag{n1}[l]{Node 1}
\psfrag{n2}[l]{Node 2}
\scalebox{0.75}{\includegraphics{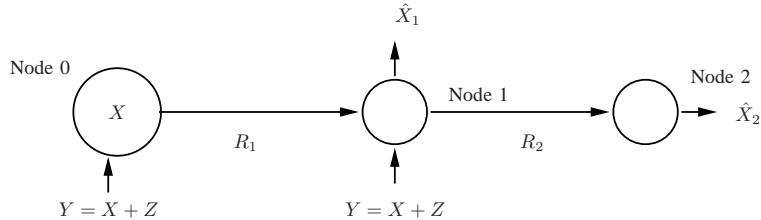}} 
\end{center}
\caption{Cascade source coding setting for the optimization problem in Corollary~\ref{coro1}. $\Xh_1$ and $\Xh_2$ are lossy reconstructions of $X$ and $Z$ is independent $X$.} \label{fig7}
\end{figure}

We now show that the setting in Figure~\ref{fig6} can be transformed into the setting in Figure~\ref{fig7}. First, we note that for the setting in Figure~\ref{fig7}, the rate distortion regions are the same regardless of whether the sources are $(X,Y)$ or $(X,\alpha Y)$ where $\alpha \neq 0$ since the nodes can simply scale $Y$ by an appropriate constant. 

Next, for Gaussian sources, the two settings are equivalent if we can show that the covariance matrix of $(X,\alpha Y)$ can be made equal to the covariance matrix of $(A+B, B)$. Equating coefficients in the covariance matrix, we require the following
\begin{align*}
\sigma_{X}^2 = \sigma_{A}^2 + \sigma_{B}^2, \\
\alpha \sigma_{X}^2 = \sigma_{B}^2, \\
\alpha^2(\sigma_{X}^2 + \sigma_{Z}^2) = \sigma_{B}^2.
\end{align*}

Solving these equations, we see that $\alpha = \sigma_{B}^2/(\sigma_{A}^2 + \sigma_{B}^2)$ and $\sigma_{Z}^2 = (\sigma_{B}^2 - \alpha^2 \sigma_{X}^2)/\alpha^2$. Since $(\sigma_{B}^2 - \alpha^2 \sigma_{X}^2) \ge 0$, this choice of $\sigma_{Z}^2$ is valid, which completes the proof.
\end{document}